\def\confversion{0}
\def\ifconf{\ifnum\confversion=1}
\def\ifnotconf{\ifnum\confversion=0}
\def\isthesis{0}

\documentclass[11 pt, fleqn, reqno]{article}
\def\showauthornotes{0}
\def\showkeys{0}
\def\showdraftbox{0}

\if\isthesis0
\usepackage{xspace,xcolor,enumerate}
\usepackage{amsmath,amssymb}
\usepackage{amsthm}
\usepackage[toc,page]{appendix}
\usepackage{thmtools}
\usepackage{thm-restate}
\usepackage{color,graphicx}
\usepackage{boxedminipage}
\usepackage{makecell}
\usepackage{tabularx}
\usepackage{enumitem}
\usepackage[linesnumbered,ruled,vlined]{algorithm2e}
\ifnum\showkeys=1
\usepackage[color]{showkeys}
\fi
\fi

\definecolor{darkred}{rgb}{0.5,0,0}
\definecolor{darkgreen}{rgb}{0,0.35,0}
\definecolor{darkblue}{rgb}{0,0,0.55}

\if\isthesis0
\usepackage[pdfstartview=FitH,pdfpagemode=UseNone,colorlinks,linkcolor=darkblue,filecolor=darkred,citecolor=darkgreen,urlcolor=darkred,pagebackref]{hyperref}
\usepackage[capitalise,nameinlink]{cleveref}
\usepackage[T1]{fontenc}
\usepackage{mathtools,dsfont,bbm}
\fi
\ifnum\showauthornotes=1
\newcommand{\Authornote}[2]{{\sf\small\color{red}{[#1: #2]}}}
\newcommand{\Authorcomment}[2]{{\sf \small\color{gray}{[#1: #2]}}}
\newcommand{\Authorfnote}[2]{\footnote{\color{red}{#1: #2}}}
\else
\newcommand{\Authornote}[2]{}
\newcommand{\Authorcomment}[2]{}
\newcommand{\Authorfnote}[2]{}
\fi

\ifnum\showdraftbox=1
\newcommand{\draftbox}{\begin{center}
  \fbox{%
    \begin{minipage}{2in}%
      \begin{center}%
        \begin{Large}%
          \textsc{Working Draft}%
        \end{Large}\\
        Please do not distribute%
      \end{center}%
    \end{minipage}%
  }%
\end{center}
\vspace{0.2cm}}
\else
\newcommand{\draftbox}{}
\fi

\if\isthesis0
\newtheorem{theorem}{Theorem}[section]

\newtheorem{definition}[theorem]{Definition}

\newtheorem{lemma}[theorem]{Lemma}
\newtheorem{remark}[theorem]{Remark}

\newtheorem{corollary}[theorem]{Corollary}
\newtheorem{claim}[theorem]{Claim}
\newtheorem{fact}[theorem]{Fact}

\theoremstyle{remark}
\newtheorem{algo}[theorem]{Algorithm}

\fi
\def\FullBox{\hbox{\vrule width 6pt height 6pt depth 0pt}}

\def\qedsketch{\ifmmode\Box\else{\unskip\nobreak\hfil
\penalty50\hskip1em\null\nobreak\hfil$\Box$
\parfillskip=0pt\finalhyphendemerits=0\endgraf}\fi}
\def\to{\rightarrow}

\def\epsilon{\varepsilon}

\def\cal{\mathcal}

\newcommand{\etal}{et al.\xspace}

\newcommand{\E}{{\mathbb E}}

\newcommand{\F}{{\mathbb F}}

\newcommand{\cB}{\mathcal{B}}

\newcommand{\cF}{\mathcal{F}}

\newcommand{\rzero}{r_{in}}
\newcommand{\rone}{r_{out}}
\newcommand{\delzero}{\delta_{in}}
\newcommand{\delone}{\delta_{out}}
\newcommand{\qzero}{q_{in}}
\newcommand{\qone}{q_{out}}

\usepackage{nicefrac}
\newcommand{\ip}[2] {\ensuremath{\left\langle #1 , #2 \right\rangle}}

\newcommand{\Esymb}{\mathbb{E}}

\makeatletter

\def\Ex#1{%
    \ProbabilityRender{\Esymb}{#1}%
}

\def\ProbabilityRender#1#2{%fancy probability command
  \@ifnextchar\bgroup%
  {\renderwithdist{#1}{#2}}
   {\singlervrender{#1}{#2}}
}
\def\singlervrender#1#2{%
   \ensuremath{\mathchoice
       {{#1}\left[ #2 \right]}
       {{#1}[ #2 ]}
       {{#1}[ #2 ]}
       {{#1}[ #2 ]}
   }
}
\def\renderwithdist#1#2#3{%
   \@ifnextchar\bgroup
   {\superfancyrender{#1}{#2}{#3}}
   {\ensuremath{\mathchoice
      {\underset{#2}{#1}\left[ #3 \right]}
      {{#1}_{#2}[ #3 ]}
      {{#1}_{#2}[ #3 ]}
      {{#1}_{#2}[ #3 ]}
     }
   }
}
\def\superfancyrender#1#2#3#4#5{
   \ensuremath{\mathchoice
      {\underset{#1}{{#1}}\left#4 #3 \right#5}
      {{#1}_{#2}#4 #3 #5}
      {{#1}_{#2}#4 #3 #5}
      {{#1}_{#2}#4 #3 #5}
   }
}
\makeatother

\newfont{\inhead}{eufm10 scaled\magstep1}

\newcommand{\calB}{{\cal B}}

\newcommand{\calO}{{\cal O}}
\newcommand{\calX}{{\cal X}}

\newcommand{\poly}{{\mathrm{poly}}}

\newcommand{\inbraces}[1]{\left\{#1\right\}}           %\inbrace{x+y}  is {x+y}
\if\isthesis1
\NewCommandCopy{\oldsection}{\section}
\NewCommandCopy{\oldsubsection}{\subsection}
\renewcommand{\section}[1]{\chapter{#1}}
\renewcommand{\subsection}[1]{\oldsection{#1}}
\fi
\DeclareSymbolFont{extraup}{U}{zavm}{m}{n}
\DeclareMathSymbol{\varheart}{\mathalpha}{extraup}{86}
\DeclareMathSymbol{\vardiamond}{\mathalpha}{extraup}{87}

\def\Ins{\mathfrak I}

\if\isthesis0

\fi

\usepackage{tikz}

\usepackage{tikz}
\usetikzlibrary{shapes.symbols}
\usepackage{float}
\usepackage{graphicx}
\usepackage{algpseudocode}
\algrenewcommand{\algorithmicrequire}{\textbf{Input:}}
\algrenewcommand{\algorithmicensure}{\textbf{Output:}}

\usepackage[top=1in, bottom=1in, left=1.25in, right=1.25in]{geometry}
\begin{document}
\sloppy

\title{List Decoding Expander-Based Codes via Fast Approximation of Expanding CSPs: I}

\author{
Fernando Granha Jeronimo\thanks{{\tt University of Illinois Urbana-Champaign}. {\tt granha@illinois.edu}. } \and
Aman Singh\thanks{{\tt University of Illinois Urbana-Champaign}. {\tt aman14@illinois.edu}. Work done toward a MS.}}

%\setcounter{page}{0}

%\date{}

\maketitle
\draftbox
\thispagestyle{empty}

\begin{abstract}
  We present near-linear time list decoding algorithms (in the block-length $n$) for expander-based code constructions. More precisely, we show that 
  \begin{itemize}
    \item[(i)] For every $\delta \in (0,1)$ and $\epsilon > 0$, there is an explicit family of good Tanner LDPC codes of (design) distance $\delta$ that is $(\delta - \epsilon, O_\varepsilon(1))$ list decodable in time $\widetilde{\mathcal{O}}_{\varepsilon}(n)$ with alphabet size $O_\delta(1)$,
    \item[(ii)] For every $R \in (0,1)$ and $\epsilon > 0$, there is an explicit family of AEL codes of rate $R$, distance $1-R -\varepsilon$ that is $(1-R-\epsilon, O_\varepsilon(1))$ list decodable in time $\widetilde{\mathcal{O}}_{\varepsilon}(n)$ with alphabet size $\exp(\poly(1/\epsilon))$, and
    \item[(iii)] For every $R \in (0,1)$ and $\epsilon > 0$, there is an explicit family of AEL codes of rate $R$, distance $1-R-\varepsilon$ that is $(1-R-\epsilon, O(1/\epsilon))$ list decodable in time $\widetilde{\mathcal{O}}_{\varepsilon}(n)$ with alphabet size $\exp(\exp(\poly(1/\epsilon)))$ using recent near-optimal list size bounds from~\cite{JMST25}.
  \end{itemize}
  
  Our results are obtained by phrasing the decoding task as an agreement CSP \cite{RWZ20,DinurHKNT19} on expander graphs and using the fast approximation algorithm for $q$-ary expanding CSPs from~\cite{Jer23}, which is based on weak regularity decomposition~\cite{JST21,FK96:focs}. Similarly to list decoding $q$-ary Ta-Shma's codes in~\cite{Jer23}, we show that it suffices to enumerate over assignments that are constant in each part (of the constantly many) of the decomposition in order to recover all codewords in the list.
\end{abstract}

\newpage

\ifnotconf
\pagenumbering{roman}
\tableofcontents
\clearpage
\fi

\pagenumbering{arabic}
\setcounter{page}{1}

\section{Introduction}

Expander graphs are fundamental objects in computer science and mathematics~\cite{HooryLW06}. They combine two seemingly contradictory properties of being well-connected while being sparse. Expanders have found applications in numerous areas~\cite{Vadhan12} such as coding theory, pseudorandomness, complexity theory, number theory, etc. In particular, their pseudorandom properties together with their sparsity have led to many important code constructions such as the Sipser--Spielman codes~\cite{SS96}, the distance amplified codes of Alon--Edmonds--Luby (AEL)~\cite{AEL95}, Alon \etal~\cite{ABNNR92}, and Ta-Shma~\cite{Ta-Shma17}, as well as, locally testable codes~\cite{DELLS22}, quantum LDPC codes~\cite{PK22} and codes with near optimal list sizes~\cite{JMST25}. While in many cases their unique decoding properties are relatively well understood now, their list decoding properties are far from fully understood. In general, a central question in coding theory is the following:
\begin{center}
  When does a family of codes admit efficient list decoding?
\end{center}
Here, we investigate this question in the context of expander-based code construction. Many such constructions are based on local-to-global phenomena, where the local properties of a constant sized code are made global using expanders. For instance, in a Tanner code construction, edges of an expander graph are associated with bits of a codeword and the local edge views at each vertex are constrained to be a codeword from a fixed code. Then, expansion ensures that the ``local'' distance of this fixed code gives rise to the ``global'' distance of the resulting code. 

In unique decoding this kind of Tanner code, one first uniquely decodes the local views and proceeds to find a globally consistent codeword from them. On the other hand, in the list decoding regime, the local views can no longer be locally decoded to a single codeword, but rather one is forced to consider local lists. The challenge is to glue these local lists to compose the global codewords in the global list. In the case of 
list decoding Tanner codes against erasures, Wootters \etal~\cite{RWZ20} use 
a propagation rounding strategy reminiscent of approximating a constraint satisfaction problem (CSP) on a natural agreement CSP whose constraints enforce that local views are consistent. With this simpler propagation strategy alone, it was then unclear how to extend the results to errors.

One notable fact about expanders is that approximating CSPs on them can be done efficiently~\cite{BarakRS11, AJT19,Jer23} provided sufficient expansion is available. This allows for more sophisticated algorithms for approximating and understanding the structure of the agreement CSPs arising from the list decoding task. In particular, we can use the weak regularity framework for CSPs from~\cite{Jer23} that runs in near-linear time. Weak regularity is a powerful graph partition technique developed by Frieze and Kannan~\cite{FK96:focs} originally conceived to approximate dense CSPs, and it has now become a central tool in the study of regularity lemmas~\cite{RodlS10}.

The larger is the alphabet size of a CSP, the stronger is the expansion requirement in order to obtain a good approximation from these algorithms. By choosing local codes whose
list sizes are independent of the block-length, we can ensure that the alphabet size remains ``constant'' while we are allowed to increase the expander degree and strengthen 
the expansion. This decoupling of alphabet size and expansion is crucial. Moreover, imposing a stronger expansion requirement does not degrade the parameters of some expander-based constructions such as AEL and Tanner codes.

Roughly speaking, the weak regularity decomposition is a suitable partition of the variables of the CSP into a constant number of parts (depending only on the approximation error). In the case of list decoding Ta-Shma's code, it was shown in~\cite{Jer23} that it suffices to consider assignments in which each part of the weak regularity partition receives a 
single alphabet value.  Here, a similar behavior is established for other expander-based code constructions. In particular, this also implies a constant bound on number of codewords in the list, which becomes more interesting beyond the trivial Johnson bound.

Our first main result is a near-linear time list decoding algorithm for Tanner codes up to the designed distance. More precisely, we have the following statement.

\begin{theorem}[Main I - Informal version of \cref{thm:tanner-main}]
    Let $q \ge 2$ be a prime power and $\delta_{0} \in (0, 1-1/q)$.
    For any $\epsilon > 0$, there exists an explicit family of Tanner codes $\{C_n \subseteq \F_q^n \}_{n}$ such that their
    \begin{enumerate}
        \item rate is $\geq 1-2H_q(\delta_{0}) - \varepsilon$,
        \item design distance is $\delta_0^2 - \varepsilon^{\Theta(1)}$, and
        \item it is $(\delta_{0}^2-\varepsilon, 2^{2^{\varepsilon^{-\Theta(1)}}})$ list decodable with a $2^{2^{\varepsilon^{-\Theta(1)}}}  \cdot \widetilde{\mathcal{O}}( n)$ time (probabilistic) decoding algorithm.
    \end{enumerate}
\end{theorem}

Our second main result is a near-linear time list decoding algorithm for the AEL codes up to capacity
\if\isthesis0\footnote{We note that the list size can be improved, but we leave the details for a future version. For instance, using an outer code that has good list recovery properties leads to a smaller lists.}\fi.

\begin{theorem}[Main II - Informal version of \cref{thm:ael-main}]
    Let $\qzero \ge 2$ be an integer and $\delzero \in (0, 1-1/\qzero)$.
    For any $\epsilon > 0$, there exists an explicit family of AEL codes $\{C_n \subseteq \F_q^n \}_{n}$ such that their
    \begin{enumerate}
        \item alphabet size is $q=\qzero^{\varepsilon^{-\Theta(1)}}$,
        \item rate is $\geq 1-H_{\qzero}(\delta)-\varepsilon$,
        \item design distance is $\delzero - \varepsilon^{\Theta(1)}$,
        \item it is $(\delzero-\varepsilon, 2^{2^{\varepsilon^{-\Theta(1)}}})$ list decodable with a $2^{2^{\varepsilon^{-\Theta(1)}}} \cdot \widetilde{\mathcal{O}}( n)$ time (probabilistic) decoding algorithm.
    \end{enumerate}
\end{theorem}

\begin{remark}
  Using an alphabet of size $q=\exp(\exp(\poly(1/\epsilon))))$ and invoking the near-optimal list size of $O(1/\epsilon)$ for AEL from~\cite{JMST25}, the list size in the theorem
  above can be further improved to near-optimal $O(1/\epsilon)$.
\end{remark}

This allows us to obtain near-linear time list decoding with bounded list sizes for capacity achieving codes in the following widely studied capacity regime~\cite{GuruswamiR06}
with near-optimal list sizes $O(1/\epsilon)$.

\begin{corollary}
    \label{cor:ael-opt-list}
    Let $R \in (0,1)$.
    For any $\epsilon > 0$, there exists an explicit family of AEL codes $\{C_n \subseteq \F_q^n \}_{n}$ such that their
    \begin{enumerate}
        \item alphabet size is $q=2^{2^{\varepsilon^{-\Theta(1)}}}$,
        \item rate is $\geq R - \varepsilon$ while distance $\geq 1-R$ (near MDS), and
        \item it is $(1-R-\varepsilon, O(1/\epsilon))$ list decodable with a $\widetilde{\mathcal{O}}_{\varepsilon}(n)$ time (probabilistic) decoding algorithm.
    \end{enumerate}
\end{corollary}

\noindent \textbf{Related Work}: The first efficient list decoding algorithm for Tanner codes against errors was given in \cite{JST23} using the Sum-of-Squares (SOS) SDP hierarchy. There, it was shown that the distance of a Tanner code can be transformed into a SOS proof which was used in a SOS list decoding framework from~\cite{AJQST19}. Also in~\cite{JST23}, analogous ideas were used in the list decoding of AEL codes. In both cases, the list decoding radius from~\cite{JST23} was up to the Johnson bound and the running time was of the form $O(n^{\poly(1/\epsilon)})$. The first list decoding algorithm for AEL up to capacity (implied by approaching the generalized Singleton bound) was given in~\cite{JMST25} using SOS proofs also having (polynomial) running time of the form $O(n^{\poly(1/\epsilon)})$.

The connection between CSPs and list decoding has appeared in many works. As mentioned above, Wootters \etal~\cite{RWZ20} used a (simple) propagation round procedure for (implicit) CSPs arising in the list decoding of erasures in the setting of Tanner codes. In~\cite{DinurHKNT18}, Dinur \etal used approximation algorithms for unique games (UG) on expanders in order to list decode direct product code constructions. Their UG instances are also agreement CSPs. In~\cite{AJQST19}, a SOS list decoding framework leveraged the CSP perspective from \cite{AJT19}. In particular, this framework handles direct sum codes whose decoding naturally corresponds to a $k$-XOR CSP. In~\cite{JQST20}, this framework was carefully applied to the decoding of Ta-Shma's breakthrough near-optimal $\varepsilon$-balanced binary codes using the CSP perspective. Subsequently, the algorithmic list decoding radius for Ta-Shma's codes was improved up to the Johnson bound by Richelson and Roy using SOS distance proofs in \cite{RR23} and the framework of~\cite{AJQST19}, and the running time is again of the form $O(n^{\poly(1/\epsilon)})$.

A new weak regularity decomposition for $k$-XOR was developed in~\cite{JST21} to give a near-linear time algorithm for Ta-Shma's code. In~\cite{Jer23}, the weak regularity was extended to
general $k$-CSPs (in order to give a near-linear time CSP algorithm analogous to the SOS results in~\cite{AJT19}) and also to provide near-linear time list decoding algorithms for $q$-ary versions of Ta-Shma's codes. There it was shown that it suffices to enumerate over CSP assignments that are constant on each factor in order to recover the list. 

This line of work started with a CSP perspective on list decoding. Subsequently, a seemingly stronger SOS distance certificate perspective 
emerged. Now, with this work, we are back to a CSP perspective for expander-based codes~\footnote{The list decoding of Ta-Shma's code is more delicate since the quality of the CSP approximation and the code parameters race with each other.}. The simplicity of this new CSP
perspective on expanders is quite illuminating in understanding what enables efficient list decoding.

\noindent \textbf{Concurrent Work}: In a concurrent work by Srivastava and Tulsiani~\cite{ST25}, similar list decoding results for AEL and Tanner codes were obtained also via weak
regularity techniques~\cite{JST21} and with similar parameter trade-offs.

\noindent \textbf{Future Work}: We leave the decoding tasks of other code families to future work.

\section{Preliminaries}

\subsection{Expander Graphs}

We denote bipartite graphs by $G = (L, R, E)$ where $L$ is the set of left vertices, $R$ is the set of right vertices and $E$ is the set of edges between them.
We often use the symbols $u$ and $v$ for vertices in $L$ and $R$ respectively.
$w$ however is used to refer to an arbitrary vertex.
$S \subseteq L$ and $T\subseteq R$ are frequently used notation for subsets on either sides.
In these terms, $E(S,T) = \{(u,v) \;|\; u \in S, v \in T, (u,v)\in E\}$.
Also, for any vertex $w \in L \cup R$, $N(w)$ is used to denote the set of neighboring vertices.
\begin{definition}[Expander Graphs]
    Let $G = (L, R, E)$ be a $d$-regular bipartite graph where $|L| = |R| = n$. We say that $G$ is an $(n, d, \lambda)$-expander if
    $\lambda = \max(|\lambda_2|, |\lambda_n|)$ where $\lambda_1\geq\lambda_2\geq \ldots\geq \lambda_n$ are the singular values of the bi-adjacency matrix of $G$.
\end{definition}

\begin{theorem}[Expander Mixing Lemma]
    Suppose that $G = (L, R, E)$ is an $(n, d, \lambda)$-expander. Then for any $S\subseteq L$, and $T \subseteq R$,
    $$
        \left|E(S,T) - \dfrac{d}{n}|S||T|\right| \leq \lambda\sqrt{|S||T|}
    $$
\end{theorem}
In a bipartite expander, the neighborhood $N(u)$ of any vertex $u \in L$ looks almost like a random sample of $R$ in the sense that if there is a subset $T \subseteq R$ such that $|T| = \alpha n$, then most vertices $u \in L$ will have $(\alpha\pm\varepsilon)d$ neighbors in $T$.
The same can be said about the neighborhood $N(v)$ of any vertex $v \in R$.
This property will be very useful in the following sections and hence we formalize it here.
\begin{claim}
    \label{clm:product-like-expansion}
    Let $G = (L, R, E)$ be a $(n, d, \lambda)$ bipartite expander, and $T \subseteq R$ be a subset with $|T| \geq \alpha n$. Then, the set of vertices in $L$ which have at least $(\alpha-\varepsilon) d$ neighbors in $T$, i.e., $S = \{u\;|\;u\in L, |N(u)\cap T| \geq (\alpha-\varepsilon) d\}\text{, is such that } |S|>\left(1-(\lambda/d)^2\cdot1/\varepsilon^2\right)n$.
    Similarly, if $S\subseteq L$ and $|S|\geq \alpha n$ then $|T| = |\{v\;|\;v\in R, |N(v)\cap S| \geq (\alpha-\varepsilon) d\}| > \left(1-(\lambda/d)^2\cdot1/\varepsilon^2\right)n$.
\end{claim}
\begin{proof}
    Let $\overline{S} = L\setminus S$ and $\overline{T} = R\setminus T$. Now, we only need to analyze $\left|E(\overline{S}, \overline{T})\right|$ using the expander mixing lemma to prove the claim. Notice that from the definition, vertices in $\overline{S}$ must have $> \left(1-(\alpha-\varepsilon)\right) d$ neighbors in $\overline{T}$. So,
\begin{align*}
    \left(1-(\alpha-\varepsilon)\right) d \cdot \left|\overline{S}\right|\;<\;\left|E\left(\overline{S}, \overline{T}\right)\right| &\;\leq\; \dfrac{d}{n} \left|\overline{S}\right| \left|\overline{T}\right| + \lambda \sqrt{\left|\overline{S}\right|\left|\overline{T}\right|} \\
    \left(1-(\alpha-\varepsilon)\right) d \cdot \left|\overline{S}\right| &\;<\; (1-\alpha) d \cdot \left|\overline{S}\right| + \lambda\sqrt{\left|\overline{S}\right| n} \\
    \left|\overline{S}\right| &\;<\; (\lambda/d)^2\cdot (1/\varepsilon^2)\cdot n
\end{align*}
And equivalently, $|S|>\left(1-(\lambda/d)^2\cdot1/\varepsilon^2\right)n$. The other case can also be proved similarly.
\end{proof}

\subsection{Codes}

We now define some coding theory concepts.
\begin{definition}[Code, distance and rate]
    A $[n,\delta, r]_q$ code $C$ over an alphabet of size $q$ with distance $\delta$, rate $r$ and block length $n$ is a subset $C \subseteq [q]^n$ such that
    $$
        r = \frac{\log_q{|C|}}{n} \quad\quad\quad \delta = \frac{\min_{f, g \in C, f \neq g} \Delta(f, g)}{n}
    $$
    We say that a family of codes, $\{C_n\}_n$, has rate $r$ and distance $\delta$ if $r(C_n) \ge r$ and $\delta(C_n) \ge \delta$ for every $n$.
\end{definition}
\begin{definition}[List decodable code]
    A code $C$ with block length $n$ is $(\alpha, \ell)$ list-decodable if at every $\widetilde{y}$ in the ambient space, the list $\{z \;|\; z \in C, \Delta(z, \widetilde{y}) \leq \alpha n\}$ has size at most $\ell$. 
\end{definition}
For the $d$-regular graph based codes that we discuss below, it is convenient to talk about the local view induced at a vertex $w$ of the graph.
So, let  
$$ z_{w} = (z_{w,N_1(w)}, z_{w,N_2(w)}, \ldots, z_{w,N_d(w)}) $$
where the neighbor functions $N_1, N_2, \ldots, N_d$ specify an arbitrary order over the $d$-sized neighbor set $N(w)$ of $w$.
The notation $z_{u,v}$ is used to indicate the alphabet associated with the edge along $v$ at $u$.
Note that in this notation $z_{u,v} = z_{v,u}$.
\begin{definition}[AEL Code]
    An AEL code $C^{\textnormal{AEL}} = C(C_{in}, C_{out}, G) \subseteq [\qzero^d]^n$ where $G = (L,R,E)$ is a $(n, d, \lambda)$ bipartite expander, $C_{in} \subseteq [\qzero]^d$ and $C_{out} \subseteq [\qone]^n$ with $|C_{in}| = \qone$.
    Let $\textnormal{Enc} : [\qone] \to [\qzero]^d$ be the encoder for the inner code $C_{in}$, extended vertex-wise on $L$ to $\textnormal{Enc} : [\qone]^n \to [\qzero]^{nd}$ by
    $$\textnormal{Enc}(x) = (\textnormal{Enc}(x_u))_{u \in L} \quad \text{for } x \in [\qone]^n$$
    For each edge $(u, v) \in E$, define the label on edge $(u,v)$ by $\textnormal{Enc}(x)_{u,v} = \textnormal{Enc}(x_u)_v$.
    Then, the AEL code is given by
    $$ C^{\textnormal{AEL}} = \left\{ (z_v)_{v \in R} \;|\; z = \textnormal{Enc}(x), \; x \in C_{out} \right\}$$
    where each $z_v \in [\qzero]^d \equiv [\qzero^d]$ aggregates the incoming labels to vertex $v\in R$.
\end{definition}

\begin{fact}
    If $C_{in}$ is a $[d, \delzero, \rzero]_{\qzero}$ code, $C_{out}$ is a $[n, \delone, \rone]_{\qone}$ code and $G$ is a $(n, d, \lambda)$ bipartite expander, then $C^{\textnormal{AEL}} = C(C_{in}, C_{out}, G)$ has rate $r = \rzero\rone$ and distance $\delta$ such that $ \delta \geq \delzero - \tfrac{\lambda/d}{\delone}$. For this reason, $\delzero - \tfrac{\lambda/d}{\delone}$ is called the design distance of $C^{\textnormal{AEL}}$.
\end{fact}
\begin{definition}[Tanner Code]
    A Tanner code $C^{\textnormal{TAN}} = C(C_0, G) \subseteq [q]^{nd}$ where $G$ is a $(n, d, \lambda)$ bipartite expander and $C_0 \subseteq [q]^d$ is a local code.
    The code is defined as 
    $$C^{\textnormal{TAN}} = \{z \;|\; z\in [q]^{nd}, \forall w\in L\cup R \;:\, z_w \in C_0\}$$
\end{definition}
\begin{fact}
    If $C_0$ is a $[d, \delta_0, r_0]_q$ code and $G$ is a $(n, d, \lambda)$ bipartite expander, then $C^{\textnormal{TAN}} = C(C_0, G)$ has rate $r \geq 2r_0-1$ and distance $\delta$ such that $\delta \geq \delta_0 (\delta_0 - \tfrac{\lambda}{d})$. For this reason, $\delta_0 (\delta_0 - \tfrac{\lambda}{d})$ is called the design distance of $C^{\textnormal{TAN}}$.
\end{fact}

\subsection{Constraint Satisfaction Problems}

\begin{definition}[2-CSP]
    An instance $\Ins = (V, D, \Psi)$ of a $2$-CSP (Constraint Satisfaction Problem) on $n$ variables $V$ over a domain $D = [\ell]$ is specified by a set of constraints $\Psi$ over tuples in $W \subseteq [n]^2$ where $\Psi = \{\psi_{(i, j)}\;|\;(i,j) \in W\}$ such that $\psi_{(i, j)} \subseteq [\ell]^2$. An assignment $b : V\to D$ is a function indicating which domain value is set at each variable. A constraint $\psi_{(i, j)}$ is said to be satisfied by an assignment $b$ if $(b(i), b(j)) \in \psi_{(i,j)}$.
    The total number of constraints satisfied by an assignment is denoted as $\textnormal{val}(\Ins, b)$ where 
    $$\textnormal{val}(\Ins, b) = \left|\left\{(i,j) \;|\; (b(i), b(j)) \in \psi_{(i,j)}, \;\psi_{(i,j)}\in \Psi\right\}\right|$$
\end{definition}

\subsection{Weak Regularity Decomposition}

For functions $f,g$ over $\calX$, $\langle f, g\rangle = \sum_{x\in\calX} f(x)g(x)$ represents inner product over the counting measure.
Also, we use the notation where $\langle f, g\rangle_W = \sum_{x\in W} f(x)g(x)$ if $W \subseteq \calX$.
The following theorem is the main external result we use. Note that we use the notation $\mathcal{\widetilde{O}}(t(n))$ to hide the polylogarithmic factors in $t(n)$.
\begin{theorem}[Simplified form of Efficient Weak Regularity \cite{Jer23}]
    \label{thm:weak-reg}
    Let $E \subseteq[n]^2$ be the set of edges of a $(n,d,\lambda)$ bipartite expander. Also suppose $\mathcal{F} = \text{CUT}^{\;\otimes2}$ is the set of cut functions, i.e., $\{\mathbf{1}_S \otimes \mathbf{1}_T \;|\; S, T \subseteq [n]\}$ and $g : [n]^2 \to \{0,1\}$ be such that it is supported on $E$ with $\langle g, g\rangle \leq |E|$. Then for every $\gamma > 0$, if $\lambda/d < \gamma^2/2^{23}$, there exists $h = \sum_{t=1}^{p} k_t \cdot f_t$ with $p = \mathcal{O}(1/\gamma^2)$, scalars $k_1,\ldots,k_p$ and functions $f_1,\ldots,f_p \in \mathcal{F}$ such that
    $$\max_{f\in \mathcal{F}}\left|\left\langle g - h, f\right\rangle\right| \leq \gamma \cdot nd$$ Moreover, $h$ can be found in $\widetilde{\mathcal{O}}(2^{2^{{\mathcal{O}}(1/\gamma^2)}}\cdot nd)$ time with high probability by a randomized algorithm.
\end{theorem}

It will be convenient to use the language of factors to search the
decompositions identified by regularity lemmas.

\begin{definition}[Factors, atoms, measurable functions]
Let $\calX$ be a finite set called the domain.
A factor $\calB = \{\cal{P}^{(i)}\}_i$ is a decomposition or a partition of the set $\calX = \sqcup_i \; \cal{P}^{(i)}$.
The parts of the factor ($\mathcal{P}^{(i)}$'s) are referred to as atoms.
A function $f$ on $\calX$ is said to $\calB$-measurable if it is constant over each atom of $\calB$.
\end{definition}
The factors we will consider will be defined by a finite collection of cut functions.
\begin{definition}[Factors from cut functions]
  Let $\calX$ be a finite set, and let $\cF = \inbraces{f_1, \ldots, f_r}$ be a finite collection of cut functions ($f_i:\calX \to \{0,1\}$).
  We consider the factor $\cB$ defined by the functions in $\cF$, as the factor with the atoms $\inbraces{x ~|~ f_1(x) = a_1, \ldots, f_p(x) = a_p}$ for all $(a_1,\ldots, a_p) \in \{0,1\}^p$.
\end{definition}
Additionally, we will need the following simple observation regarding conditional averages.
\begin{definition}[Conditional averages]
If $f$ is a function and $\calB$ is a factor, then we define the  conditional average function $\Ex{f|\calB}$ as
$$
\Ex{f|\calB}(x) ~=~ \Ex{y \sim \calB(x)}{f(y)}
$$
where $\calB(x)$ denotes the atom containing $x$. Note that the function $\Ex{f|\cB}$ is measurable with respect to $\cB$.
\end{definition}

\begin{lemma}\label{lem:measurable-inner-product}
  Let $h$ be a function over the domain $\calX$ that is $\calB$-measurable, and $f$ be an arbitrary function over the same domain. Then, for the counting measure over $\calX$, we have
  $$
    \ip{h}{f} ~=~ \ip{h}{\Ex{f|\calB}} 
  $$
\end{lemma}
\begin{proof}
By definition of the $\calB$-measurability, $h$ is constant on each atom, and thus we can write $h(x)$ as $h(\calB(x))$. 
  \begin{align*}
    \ip{h}{f} ~=~ \sum_{x \in \calX}{h(x) \cdot f(x)}
    &~=~ \sum_{x \in \calX} \frac{1}{|\calB(x)|} \sum_{y \sim \calB(x)}{h(y) \cdot f(y)} \\
    &~=~ \sum_{x \in \calX}{h(\calB(x)) \cdot \frac{1}{|\calB(x)|} \sum_{y \sim \calB(x)}{f(y)}} \\
    &~=~ \sum_{x \in \calX}{h(x) \cdot \Ex{f|\calB}(x)} ~=~ \ip{h}{\Ex{f|B}} \qedhere 
  \end{align*}
\end{proof}

\begin{definition}[Restricted factors]
    The restriction of a factor $\mathcal{B}$ on a subset $W$ of the domain $\mathcal{X}$ is denoted as $\mathcal{B}_{|W} = \{\mathcal{P} \cap W\;|\; \mathcal{P} \in \mathcal{B}\}$.
\end{definition} 
Note $\mathcal{B}_{|W}$ is itself a factor as it is a partition of the restricted domain $W$. Also, if a function is $\calB$-measurable it is also $\calB_{|W}$-measurable.
\begin{definition}[Concentrated functions]
    A function is said to be $\eta$-concentrated on a factor $\mathcal{B}$ over domain $\mathcal{X}$ if it is measurable on the factor $\mathcal{B}_{|W}$ (over $W$) for some $W\subseteq \mathcal{X}$ with $|W|\geq(1-\eta)|\mathcal{X}|$.
\end{definition}

\section{Technical Overview}

In the expander-based codes we study, the local views of the codewords on vertices of the expander graph belong to a base code. For these codes, the list decoding task at a corrupted codeword $\widetilde{y}$ has two different steps.
In step $(1)$, we retrieve a list of (say $\ell$) candidate local codewords at each vertex.
This step is straightforward as it just involves list decoding around the local view $\widetilde{y}_u$ at each vertex $u$ according to a local code with a constant block length.
Say that after this step, we have $\Theta(n)$ different $\ell$-sized lists of local codewords where $n$ is the block length of the final code.
Now notice that for any codeword $z$ that is close to $\widetilde{y}$, the local views $z_u$ must also be close to $\widetilde{y}_u$ at most vertices $u$ in the graph.
This observation leads us to a naive list decoding algorithm where, in step $(2)$, we enumerate over the $\ell^{\Theta(n)}$ possible choices of local codewords and hence land close to every $z$ we are interested in.
This enumeration does not directly give us every $z$ because at some vertices the local view $z_u$ can be far away from $\widetilde{y}_u$, i.e. local list at $u$ may not contain $z_u$.
Even in such cases, some $\widetilde{z}$ that is enumerated over will be close enough that standard unique decoding algorithms are able to recover $z$ from $\widetilde{z}$.

The naive way of executing step $(2)$ is inefficient as it requires us to go over exponentially many choices.
This is also redundant as most local views in codewords close to $\widetilde{y}$ must be consistent with each other.
So we can rule out many inconsistent combinations of local views and save time by avoiding checking them.
The main technique used in the paper essentially only checks a few (actually constant) relevant combinations of choices.
This enables step $(2)$ to run in $\mathcal{\widetilde{O}}(n)$ time.
On a high level, this is done by encoding the consistency checks in a CSP and then using the weak regularity decomposition described in \cite{Jer23} to retrieve a ``simple'' proxy to the CSP. 
This ``simple'' proxy can then be used to efficiently enumerate the space of good solutions to the CSP which contains the list we want to recover.\\
\\
\textbf{Efficient Weak Regularity Decompositions}.
We take a detour to recall how the weak regularity framework can be used to understand the value of certain CSPs.
We will do this by describing how the efficient weak regularity framework of \cite{Jer23} works within a particular set of parameters.
Essentially, for the case where $E \subseteq [n]^2$ is derived from the edges of an expander and $g : E \to \{0,1\}$ is an arbitrary function, the framework will allow us to find a ``simple'' function $h$ that closely approximates $g$.
The notion of approximation is with respect to a set of test functions $\mathcal{F}$ and an inner product $\langle\cdot,\cdot\rangle$ on the space of functions, where, given some approximation error $\gamma > 0$, we can find $h$ such that
$$\max_{f \in \mathcal{F}}\; \left|\langle g-h, f\rangle\right| \leq \gamma$$
The sense in which $h$ is ``simple'' is that there exists scalars $k_t$ and functions $f_t \in \mathcal{F}$ such that $h = \sum_{t=1}^{p} k_t\cdot f_t$ where the number of functions in the sum is a constant $p = \mathcal{O}(1/\gamma^2)$.
This decomposition is algorithmically useful too as \cref{thm:weak-reg} states that the functions $f_t$ and constants $k_t$ can be found efficiently. 

In our particular example, we will take $\mathcal{F}$ to be the set of cut functions $\text{CUT}^{\otimes 2} = \{\mathbf{1}_S \otimes \mathbf{1}_T : [n]^2 \to \{0,1\}\;|\; S, T \subseteq[n]\}$ where $\mathbf{1}_W(w) = w\in W$ and use the counting measure where $$\langle f, g\rangle = \sum_{(u,v)\in [n]^2} f(u,v)g(u,v)$$ 
We will now consider a 2-CSP $\Ins$ where the variables are two distinct sets of $[n]$, the constraints are over $E \subseteq [n]^2$ and the domain is $[\ell]$.
In this setup, the number of constraints satisfied by an assignment $b$ is denoted as $\text{val}(\Ins, b)$.
To express this value in a way amenable to the weak regularity framework, we will have to define some more terms.
Let $g_{(\alpha_L, \alpha_R)}:[n]^2\to\{0,1\}$ be a cut function such that $g_{(\alpha_L, \alpha_R)}(u,v) = 1$ if $(\alpha_L, \alpha_R)$ satisfies the $(u,v) \in E$ constraint and $0$ otherwise.
Also, let $\mathbf{1}_{b, \alpha}$ denote a cut function where $\mathbf{1}_{b, \alpha}(w) = 1$ iff $b(w) = \alpha$.
Then,
$$
    \text{val}(\Ins, b) = \sum_{\alpha =(\alpha_L, \alpha_R) \in [\ell]^2} \langle g_{\alpha}, \mathbf{1}_{b, \alpha_L} \otimes \mathbf{1}_{b, \alpha_R} \rangle
$$
We can now approximate each function $g_{\alpha}$ by a simpler $h_{\alpha}$ with precision $\gamma$ so that 
$$
\max_{f \in \text{CUT}^{\otimes 2}}\; \left|\langle g_{\alpha}-h_{\alpha}, f\rangle\right| \leq \gamma
$$
Observe that the functions $h_\alpha = \sum_{t=1}^{p_\alpha} k_{\alpha,t} \cdot \mathbf{1}_{\mathcal{S}_{\alpha, t}} \otimes \mathbf{1}_{\mathcal{T}_{\alpha, t}}$ can be used as a simple proxy to approximate the whole CSP's value, as $\mathbf{1}_{b, \alpha_L} \otimes \mathbf{1}_{b, \alpha_R} $ belongs to the set of fooled test functions $\text{CUT}^{\otimes 2}$ too.
\begin{align*}
    \text{val}(\Ins, b) &= \sum_{\alpha \in [\ell]^2} \langle h_{\alpha}, \mathbf{1}_{b, \alpha_L} \otimes \mathbf{1}_{b, \alpha_R} \rangle \;\pm\; \gamma\ell^2\\
    &= \sum_{\alpha \in [\ell]^2}\sum_{t=1}^{p_\alpha} k_{\alpha,t} \langle \mathbf{1}_{\mathcal{S}_{\alpha, t}} \otimes \mathbf{1}_{\mathcal{T}_{\alpha, t}}, \mathbf{1}_{b, \alpha_L} \otimes \mathbf{1}_{b, \alpha_R} \rangle \;\pm\; \gamma\ell^2 \\
    &= \sum_{\alpha \in [\ell]^2}\sum_{t=1}^{p_\alpha} k_{\alpha,t} \langle \mathbf{1}_{\mathcal{S}_{\alpha, t}}, \mathbf{1}_{b, \alpha_L}\rangle\langle \mathbf{1}_{\mathcal{T}_{\alpha, t}},  \mathbf{1}_{b, \alpha_R} \rangle \;\pm\; \gamma\ell^2
\end{align*}
Ignoring the approximation error $\gamma\ell^2$ which can be made small, the form of the final equation above implies that $\text{val}(\Ins, b)$ is just controlled by the following set of values
$$
\left\{\langle \mathbf{1}_{\mathcal{S}_{\alpha, t}}, \mathbf{1}_{b, \alpha_L}\rangle\right\}_{\alpha\in[\ell]^2, t\in[p_\alpha]} 
\quad\cup\quad
\left\{\langle \mathbf{1}_{\mathcal{T}_{\alpha, t}}, \mathbf{1}_{b, \alpha_R}\rangle\right\}_{\alpha\in[\ell]^2, t\in[p_\alpha]}
$$
Notice that the size of this set is $\mathcal{O}(\ell^2/\gamma^2)$ which is a constant.
To understand this decomposition a bit better, build factors such that for every $\mathcal{P}\in\mathcal{B}_L$ or $\mathcal{P}\in\mathcal{B}_R$, the functions $\{\mathbf{1}_{\mathcal{S}_{\alpha, t}} \}_{\alpha\in[\ell]^2, t\in[p_\alpha]}$ or $\{\mathbf{1}_{\mathcal{T}_{\alpha, t}} \}_{\alpha\in[\ell]^2, t\in[p_\alpha]}$ are all constant on $\mathcal{P}$ respectively.
The number of such parts is again a constant $2^{\mathcal{O}(\ell^2/\gamma^2)}$.
It directly follows that $\text{val}(\Ins, b)$ is controlled by the set of values
$$
\Big\{\langle \mathbf{1}_{\mathcal{P}}, \mathbf{1}_{b, \alpha}\rangle= \text{number of } w\in \mathcal{P} \text{ where } b(w) = \alpha\Big\}_{\mathcal{P}\in\mathcal{B}_L \cup \mathcal{B}_R, \alpha\in[\ell]}
$$
But, partitioning the input space as such allows us to see that the CSP's value is only dependent on the distribution over $[\ell]$ inside each independent part.
Through a more careful analysis, it can also be shown that the number of satisfied constraints in any sub-portion of the CSP is also only dependent on the distributions inside the parts.\\ 
\\
\textbf{Efficiently enumerating stitched local views}.
Let's go back to the list decoding problem for expander-based codes, where we had $\ell$ candidate codewords at each vertex that need to be efficiently stitched together.
We will do this by constructing a CSP where the variables are on the vertices of the graph and the choice of its value determines which candidate codeword is chosen.
The constraints are consistency checks along the edges of the graph.
The assignments to the CSP which correspond to codewords within the list decoding radius satisfy many of these consistency checks and thus constitute a good assignment.
We can also show that some sub-portion of the CSP is fully satisfied by such assignments, i.e., for some $S,T \subseteq [n]$ if $\text{val}_{S,T}(\Ins, b)$ represents the number of satisfied $(u,v)\in E(S,T)$ constraints, then
$$\text{val}_{S,T}(\Ins, b) = |E(S, T)|$$
This observation is interesting as, if we change the assignment to some $b'$ that differs from $b$ at a few vertices, we will start dissatisfying a lot of connected constraints in $E(S,T)$.
The reason for this is the fact that changing assignments corresponds to choosing different local codewords with significantly different opinions on the edges.
This will automatically start failing the consistency checks as all were satisfied before the change.
But looking at the value of this fully satisfied sub-portion of the CSP from weak regularity's perspective, the value should not vary more than the slack granted by the approximation factor $\gamma$ if the assignment change is made carefully so as to not alter the distribution induced by $b$ on any part.
As $\gamma$ can be made extremely small, this must mean that the wiggle room for changing assignments while keeping the parts' distributions the same must also be small.
This can only happen if there is mostly just one option to choose from in each part, i.e., the distribution induced by $b$ is almost delta function like on all the parts.

We can exploit the property of good assignments inducing delta function like distributions on the parts (i.e. being $\eta$-concentrated) to efficiently enumerate the space and cover all the codewords within the list decoding radius.
Recall that the number of parts is a constant $2^{\mathcal{O}(\ell^2/\gamma^2)}$.
So, enumerating all the combinations of delta function distributions on parts will just require a constant $\ell^{2^{\mathcal{O}(\ell^2/\gamma^2)}}$ number of iterations.
This enumeration again lands us close enough to each codeword we want to recover such that a unique decoder can complete the job.

\usetikzlibrary{shapes.geometric, arrows.meta, positioning}

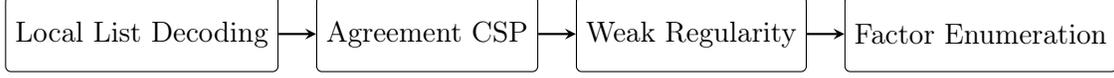
\begin{figure}[h!]
  \centering
\tikzstyle{block} = [
  rectangle,
  rounded corners=2pt,
  minimum width=1.5cm,
  minimum height=1cm,
  text centered,
  draw=black,
  line width=0.4pt
]
\tikzstyle{arrow} = [->, thick, >=stealth]

\begin{tikzpicture}[node distance=0.5cm and 0.5cm]

% Nodes
\node (local)   [block] {Local List Decoding};
\node (csp)     [block, right=of local] {Agreement CSP};
\node (weakreg) [block, right=of csp] {Weak Regularity};
\node (enum)    [block, right=of weakreg] {Factor Enumeration};
%\node (stitch)  [block, right=of enum] {Global Stitching};

% Arrows
\draw [arrow] (local) -- (csp);
\draw [arrow] (csp) -- (weakreg);
\draw [arrow] (weakreg) -- (enum);
%\draw [arrow] (enum) -- (stitch);
\end{tikzpicture}
\caption{List decoding pipeline via expanding CSPs.}
\end{figure}

\section{List Decoding \text{AEL} Codes}
Consider an \text{AEL} code $C^{\textnormal{AEL}}=C(C_{in}, C_{out}, G) \subseteq [\qzero^d]^n$ where $C_{out} \subseteq [\qone]^n$ is the outer code, $C_{in} \subseteq [\qzero]^d$ is the inner code with $|C_{in}| = \qone$ and $G = (L, R, E)$ is a $(n, d, \lambda)$ bipartite expander. Codewords in $C^{\text{AEL}}$ belong to the $[\qzero^d]^R$ space, but they also naturally induce a mapping over the edges $E \to [\qzero]$.
In this section, we will think of them as belonging to $[\qzero]^E$.
A corrupted codeword can also similarly be described as an element of $[\qzero]^E$.
We define the left and right folded views of $z \in [\qzero]^E$ as 
$$
    z^L = (z_{1_L}, z_{2_L}, \ldots, z_{n_L}) \in [\qzero^d]^L \quad\quad z^R = (z_{1_R}, z_{2_R}, \ldots, z_{n_R}) \in [\qzero^d]^R
$$
respectively.
Note that $L = \{1_L, 2_L, \ldots, n_L\}$ and $R = \{1_R, 2_R, \ldots, n_R\}$, where the numbers specify an arbitrary order on the sets of left and right vertices respectively.
In this notation, if $z\in C^{\text{AEL}}$ then $z^L \in C_{in}^L$ and 
$$\text{Enc}^{-1}(z^L) = \left(\text{Enc}^{-1}(z_{1_L}), \text{Enc}^{-1}(z_{2_L}), \ldots, \text{Enc}^{-1}(z_{n_L})\right) \in C_{out}$$
In this section, we will show that there exists a family of \text{AEL} codes that are list decodable with a radius close to the (design) distance by a near-linear time algorithm (\cref{thm:ael-main}).

This will be shown by first proving \cref{thm:AEL-listdec} and then later describing an appropriate instantiation of codes in \cref{subsec:ael-instantiate}.
\begin{theorem}
    \label{thm:AEL-listdec}
    Let $C_{in}$ be a $(\delzero, \ell)$ list decodable code and $C_{out}$ be uniquely decodable from $\delone$ fraction of errors, then the \text{AEL} code $C^{\textnormal{AEL}}=C(C_{in}, C_{out}, G)$ is $(\delzero-\varepsilon, \ell^{2^{\mathcal{O}(\ell^6/\varepsilon^2)}})$ list decodable for any constant $0<\varepsilon<\delzero$, if $\tfrac{\lambda}{d} < \tfrac{\varepsilon^2\delone^2}{2^{31}\ell^4}$. Furthermore, if $C_{out}$ is uniquely decodable from $\delone$ fraction of errors in time $t$, then this can be done in $\widetilde{\cal{O}}_{\varepsilon,\ell}(n+t)$ time.
\end{theorem}

\subsection{Setting up the CSP}
\label{subsec:ael-csp-setup}
At each vertex $u \in L$, we can use the list decoding property of $C_{in}$ to retrieve a list of $\leq \ell$ candidate local codewords at most $\delzero d$ away from the local view $\widetilde{y}_{u}$. Using these lists, we can set up a 2-CSP $\Ins_{\widetilde{y}} = (V,D,\Psi)$ with $2n$ variables, i.e., $V = \{x_w \;|\; w \in L \cup R\}$ and domain $D = [\ell]$. Each variable represents the choice of a candidate codeword from the local list at a vertex. The locally decoded list at each vertex $u \in L$ is arbitrarily ordered like $c^{1}_u, c^{2}_u, \ldots, c^{\ell}_u$ so that every choice of a domain value matches to a candidate local codeword. If there are $<\ell$ elements in the local list, the remaining domain values are set to represent arbitrary codewords in the rest of $C_{in}$. This way, the choices $c^{1}_u, c^{2}_u, \ldots, c^{\ell}_u$ are all unique. The domain values associated with vertices in $R$, however, do not represent local codewords. As we use the notation $c^{i}_{u,v}$ to represent the symbol in $c^{i}_u$ along vertex $v \in N(u)$, note that in this notation $c^{i}_{u,v} = c^{i}_{v,u}$ is not generally true.

Each edge in the graph $G$ represents a constraint in the CSP, i.e., $\Psi = \{\psi_e \; | \; e \in E\}$. At an edge $(u, v) \in E$, the assignment $x_u = i, x_v = j$ satisfies the constraint iff $c^{i}_{u,v}$ agrees with $\widetilde{y}_{u,v}$, i.e., $\psi_{(u,v)} = \{(i,j) \; | \; c^{i}_{u, v} = \widetilde{y}_{u, v}\}$. Notice that any choice of the value of variables associated with $R$ yields the same set of satisfied constraints. This implies that this CSP's value solely depends on the choice of variables associated with $L$. 

\begin{algorithm}
\caption{List Decoding of \text{AEL} code}
\label{alg:AEL}
\begin{algorithmic}
\Require $\widetilde{y} \in [\qzero]^E$
\Ensure $\mathcal{L} \subseteq C^{\text{AEL}}$
\\
\\
\begin{itemize}
    \item Construct a CSP $\Ins_{\widetilde{y}} = (V,D=[\ell],\Psi = \{\psi_e\;|\; e \in E\})$ as described in \cref{subsec:ael-csp-setup}
    \item For each $\alpha \in [\ell]^2$, obtain the weak regularity approximation $h_\alpha$ of $g_\alpha : [n]^2 \to \{0,1\}$ as defined below, with $\gamma = \dfrac{\varepsilon\delone}{2^4\ell^2}$
\[g_\alpha(u, v) = \begin{cases}
    1 &\quad \text{if }\,(u,v) \in E \,\text{ and }\,\alpha \in \psi_{(u,v)} \\
    0 &\quad \text{otherwise}
\end{cases}\]
    \item Let $h_\alpha = \sum_{t=1}^{p_\alpha} k_{\alpha,t} \cdot \mathbf{1}_{\mathcal{S}_{\alpha, t}} \otimes \mathbf{1}_{\mathcal{T}_{\alpha, t}}$. Then, compute the factor $\mathcal{B}_L$ determined by the set of cut functions $\{\mathbf{1}_{\mathcal{S}_{\alpha, t}} \;|\; \alpha\in[\ell]^2, t\in[p_\alpha]\}$ over $L$
    \item $\mathcal{L} \gets \{\}$
    \item For every $\mathcal{B}_L$ measurable function $x:L \to[\ell]$ \begin{itemize}
        \item[-] Set $\widetilde{z}$ such that $\widetilde{z}^L = \left(c^{x(1)}_1, c^{x(2)}_2, \ldots, c^{x(n)}_n\right)$
        \item[-] Set $z$ such that $z^L \gets \text{Enc}\left(\textsc{UniqueDecode}(\text{Enc}^{-1}(\widetilde{z}^L))\right)$
        \item[-] If $z \neq \textsc{Reject}$ and $\Delta(z^R, \widetilde{y}^R) \leq (\delzero - \varepsilon)n$
        \begin{itemize}
            \item[-] $\mathcal{L}\gets\mathcal{L} \cup \{z\}$
        \end{itemize}
    \end{itemize}
    \item Return $\mathcal{L}$
\end{itemize}
\end{algorithmic}
\end{algorithm}

\subsection{Weak Regularity Decomposition}
\label{subsec:ael-weak-reg}
We will now use the weak regularity decomposition to analyze the structure of the CSP. Notice that with $g_\alpha$ as defined in \cref{alg:AEL}, the number of constraints satisfied by an assignment $x : V \to [\ell]$ in the set up CSP $\Ins_{\widetilde{y}} = (V, D, \Psi)$ can be written as
$$
    \text{val}(\Ins_{\widetilde{y}}, x) = \sum_{\alpha=(\alpha_L, \alpha_R) \in [\ell]^2} \left\langle g_\alpha, \mathbf{1}_{x,{\alpha_L}} \otimes \mathbf{1}_{x,{\alpha_R}} \right\rangle
$$
where $\mathbf{1}_{x,{\alpha_L}}$ and $\mathbf{1}_{x,{\alpha_R}}$ are cut functions where $\mathbf{1}_{x,{\alpha_L}}(u) = 1$ iff $x(u) = \alpha_L$ for $u \in L$ and $\mathbf{1}_{x,{\alpha_R}}(u) = 1$ iff $x(v) = \alpha_R$ for $v \in R$. 
It is also possible to write the number of satisfied constraints of the form $\{\psi_e\;|\;e \in E(S, T)\;\}$ for any $S\subseteq L$ and $T\subseteq R$ in a similar way,
$$
    \text{val}_{S,T}(\Ins_{\widetilde{y}}, x) = \sum_{\alpha=(\alpha_L, \alpha_R) \in [\ell]^2} \left\langle g_\alpha, \mathbf{1}_{x,{\alpha_L}} \cdot \mathbf{1}_{S} \otimes \mathbf{1}_{x,{\alpha_R}} \cdot \mathbf{1}_{T} \right\rangle
$$
Using the decomposition from \cref{thm:weak-reg}, each $g_\alpha$ can be approximated by $h_\alpha$, giving us the following expression \footnote{As the CSP only depends on the variables in $L$, we can further optimize the dependence of $\ell$ in the error term. But we skip this optimization here for consistency with the other sections.}
$$
    \text{val}_{S,T}(\Ins_{\widetilde{y}}, x) = \sum_{\alpha=(\alpha_L, \alpha_R) \in [\ell]^2} \left\langle h_\alpha, \mathbf{1}_{x,{\alpha_L}} \cdot \mathbf{1}_{S} \otimes \mathbf{1}_{x,{\alpha_R}} \cdot \mathbf{1}_{T} \right\rangle \;\pm\; \ell^2 \cdot \gamma \cdot nd 
$$
Assuming $h_\alpha$ has the form $\sum_{t=1}^{p_\alpha} k_{\alpha,t} \cdot \mathbf{1}_{\mathcal{S}_{\alpha, t}} \otimes \mathbf{1}_{\mathcal{T}_{\alpha, t}}$, the factor $\mathcal{B}_L$ is determined by the set $\{\mathbf{1}_{\mathcal{S}_{\alpha, t}} \;|\; \alpha\in[\ell]^2, t\in[p_\alpha]\}$ and $\mathcal{B}_R$ by $\{\mathbf{1}_{\mathcal{T}_{\alpha, t}} \;|\; \alpha\in[\ell]^2, t\in[p_\alpha]\}$.
So,
\begin{align*}
    \text{val}_{S,T}(\Ins_{\widetilde{y}}, x) &= \sum_{\alpha \in [\ell]^2} \sum_{t=1}^{p_\alpha} k_{\alpha,t} \left\langle \mathbf{1}_{\mathcal{S}_{\alpha, t}} \otimes \mathbf{1}_{\mathcal{T}_{\alpha, t}}, \; \mathbf{1}_{x,{\alpha_L}} \cdot \mathbf{1}_{S} \otimes \mathbf{1}_{x,{\alpha_R}} \cdot \mathbf{1}_{T} \right\rangle \;\pm\; \ell^2 \gamma nd \\
    &= \sum_{\alpha, t} k_{\alpha,t} \left\langle \mathbf{1}_{\mathcal{S}_{\alpha, t}} , \mathbf{1}_{x,{\alpha_L}} \right\rangle_S\left\langle \mathbf{1}_{\mathcal{T}_{\alpha, t}}, \mathbf{1}_{x,{\alpha_R}} \right\rangle_T \;\pm\; \ell^2 \gamma nd \\
    &= \sum_{\alpha, t} k_{\alpha,t} \left\langle \mathbf{1}_{\mathcal{S}_{\alpha, t}} , \E[\mathbf{1}_{x,{\alpha_L}} \,| \,\mathcal{B}_{L|S}] \right\rangle_S\left\langle \mathbf{1}_{\mathcal{T}_{\alpha, t}}, \E[\mathbf{1}_{x,{\alpha_R}} \,|\, \mathcal{B}_{R|T}] \right\rangle_T \;\pm\; \ell^2 \gamma nd
\end{align*}
The last step is possible because all the functions $\{\mathbf{1}_{\mathcal{S}_{\alpha, t}}\}_{\alpha,t}$ and $\{\mathbf{1}_{\mathcal{T}_{\alpha, t}}\}_{\alpha,t}$ are measurable on $\mathcal{B}_{L|S}$ and $\mathcal{B}_{R|T}$ respectively.
Furthermore, due to the nature of the constraints, the value above does not depend on the value of assignment function on $R$.
So, $\text{val}_{S,T}(\Ins_{\widetilde{y}}, x) = \text{val}_{S,T}(\Ins_{\widetilde{y}}, x[R\to1])$, where $x[R\to1]$ maps the variables in $u\in L$ to $x(u)$, but $v\in R$ to a $1$. 
As $\left\langle \mathbf{1}_{\mathcal{T}_{\alpha, t}}, \E[\mathbf{1}_{x[R\to1],{\alpha_R}} \,|\, \mathcal{B}_{R|T}] \right\rangle$ is a constant independent of $x$, it can be incorporated under $k_{\alpha,t}$
$$
    \text{val}_{S,T}(\Ins_{\widetilde{y}}, x) = \sum_{\alpha, t} k'_{\alpha,t} \cdot\left\langle \mathbf{1}_{\mathcal{S}_{\alpha, t}} , \E[\mathbf{1}_{x,{\alpha_L}} \,|\, \mathcal{B}_{L|S}] \right\rangle_S \;\pm\; \ell^2 \gamma nd
$$
We can see that
$$\E[\mathbf{1}_{x,{\alpha_L}} \,|\, \mathcal{B}_{L|S}](u) = \dfrac{\# w \in \calB_{L|S}(u) \text{ where } x(w) = \alpha_L}{\left|\calB_{L|S}(u)\right|}$$
which implies that $\text{val}_{S,T}(\Ins_{\widetilde{y}}, x)$ is dependent only on the distribution induced by the assignment function $x$ on the factor $\mathcal{B}_{L|S}$.

\subsection{Concentration over Factors}
In this section, we will see how codewords within the list-decoding radius impose ``concentration'' over the factors.
Due to the fact that these codewords highly satisfy the CSP, we will prove that they must essentially be locked to a single choice within each part of the factor decomposition.
This concentrated structure imposed by such codewords will allow us to efficiently find them.
More formally, for any $z$ in the ambient space $[\qzero]^E$, define the subsets $S_z\subseteq L$ and $T_z \subseteq R$ as
\begin{align*}
    &S_z^* = \{u \;|\; u \in L,\Delta(z_u, \widetilde{y}_u)\leq \delzero d\} \\
    &S_z = \{u \;|\; u \in S_z^*, |N(u) \setminus T_z|\leq(\delzero-\varepsilon/2)d\} \quad\quad T_z = \{v \;|\; v \in R, z_v = \widetilde{y}_v\}
\end{align*}
And the assignment function $x_z : V \to [\ell]$ as
$$
    x_z(w) = \begin{cases}
    i &\quad \text{if } w \in S_z \text{ and } z_w = c^i_w\\
    \text{arbitrary} &\quad \text{otherwise}
    \end{cases}
$$
\begin{claim}    \label{clm:full-sat-ael}
    $\textnormal{val}_{S_z, T_z}(\Ins_{\widetilde{y}}, x_z) = |E(S_z, T_z)|$
\end{claim}
\begin{proof}
    The constraint associated with $(u,v)\in E$ is satisfied by $x_z$ if $c^{x_z(u)}_{u,v} = \widetilde{y}_{u,v}$. If $u \in S_z$, then $c^{x_z(u)}_{u,v} = z_{u,v}$ because $z_u$ is part of the locally decoded list for all $u \in S_z^*$. Also, by definition, $z_{u,v} = \widetilde{y}_{u,v}$ for $v\in T_z$. So, for any $(u,v) \in E(S_z, T_z), \;c^{x_z(u)}_{u,v} = z_{u,v}= \widetilde{y}_{u,v} $.
\end{proof}

\begin{lemma}
    \label{lem:ael-concentration}
    If $z \in C^{\textnormal{AEL}}$ and $\Delta(z^R, \widetilde{y}^R) \leq (\delzero - \varepsilon)nd$, then $x_z$ is $\eta$-concentrated on $\mathcal{B}_{L|S_z}$ for $\eta = (5\ell^2\gamma/\varepsilon)\cdot n/|S_z|$.
\end{lemma}
\begin{proof}
    We will prove this by contradiction. Let $z \in C^{\text{AEL}}$ such that $\Delta(z^R, \widetilde{y}^R) \leq \delzero - \varepsilon$ but $x_z$ is not $\eta$-concentrated. If this is the case, we will construct a different assignment function $x'_z$ as follows atom-wise. 
    First, order the output values $[\ell]$ of the assignment function by the number of times they are mapped to in an atom $\mathcal{P}$ breaking ties arbitrarily. This gives us a ranking function $rank : [\ell] \to [\ell]$ where $rank(1)$ is a most frequently mapped value and $rank(\ell)$ is a least frequently one. For each $i \in [\ell]$ define $\mathcal{P}_{i} = \{w \;|\; w\in\mathcal{P}, x_z(w) = rank(i)\}$.
    Now, set $x'_z(w) = rank(i-1)$ if $ w \in \mathcal{P}_i$ and $i \in [2\ldots\ell]$. The elements in $\mathcal{P}_1$ are assigned values so that the distribution inside $\mathcal{P}$ does not change as a whole. 
    Specifically, an arbitrary set of $|\mathcal{P}_{i-1}|-|\mathcal{P}_{i}|$ elements get mapped to $rank(i-1)$ for each $i \in [2\ldots\ell]$ and $|\mathcal{P}_\ell|$ elements get mapped to $rank(\ell)$. Notice that this makes $x'_z$ and $x_z$ differ on at least $|\mathcal{P}_2 \sqcup \cdots \sqcup \mathcal{P}_\ell|$ elements for each atom $\mathcal{P}$. We make the following two observations now

    \begin{enumerate}
        \item As $x_z$ is not $\eta$-concentrated, it must not map $\geq \eta |S_z|$ (\textit{bad}) elements to their corresponding atom's majority value. This implies $x_z$ and $x'_z$ must differ on these elements. For such \textit{bad} $u \in S_z$, $\Delta(c^{x_z(u)}_u, c^{x'_z(u)}_u) \geq \delzero d$ but only $\leq(\delzero-\varepsilon/2)d$ of its neighbors are not in $T_z$ by the definition. So, $x_z$ and $x_z'$ differ on $\geq (\varepsilon/2)d\cdot \eta |S_z|$ constraints. As $x_z$ satisfies all constraints in $E(S_z, T_z)$ (\cref{clm:full-sat-ael}), any such difference can only lead to constraints becoming unsatisfied, implying
        $$
        \text{val}_{S_z,T_z}(\Ins_{\widetilde{y}}, x'_z) \;\leq\; \text{val}_{S_z,T_z}(\Ins_{\widetilde{y}}, x_z) - (\eta\varepsilon/2) \cdot |S_z|d
        $$
        \item $x'_z$ and $x_z$ have the same distribution over $[\ell]$ in any atom in $\mathcal{B}_L$, so the weak regularity decomposition tells us that
        $$
            \left|\text{val}_{S_z,T_z}(\Ins_{\widetilde{y}}, x_z) - \text{val}_{S_z,T_z}(\Ins_{\widetilde{y}}, x'_z)\right| \;\leq\; 2 \ell^2\gamma\cdot nd
        $$
    \end{enumerate}
    The above observations lead to a contradiction as $(\eta\varepsilon/2) \cdot |S_z|d > 2 \ell^2\gamma\cdot nd$ when $\eta = (5\ell^2\gamma/\varepsilon)\cdot n/|S_z|$.
\end{proof}

\subsection{Covering relevant codewords}
We will now prove that concentration over the factors enables the enumeration described in \cref{alg:AEL} to cover all the relevant codewords (the ones within the list decoding radius). 

\label{subsec:ael-finish}
\begin{theorem}[\cref{thm:AEL-listdec} restated]
    Let $C_{in}$ be a $(\delzero, \ell)$ list decodable code and $C_{out}$ be uniquely decodable from $\delone$ fraction of errors, then the \text{AEL} code $C^{\textnormal{AEL}}=C(C_{in}, C_{out}, G)$ is $(\delzero-\varepsilon, \ell^{2^{\mathcal{O}(\ell^6/\varepsilon^2)}})$ list decodable for any constant $0<\varepsilon<\delzero$, if $\tfrac{\lambda}{d} < \tfrac{\varepsilon^2\delone^2}{2^{31}\ell^4}$. Furthermore, if $C_{out}$ is uniquely decodable from $\delone$ fraction of errors in time $t$, then this can be done in $\widetilde{\cal{O}}_{\varepsilon,\ell}(n+t)$ time.
\end{theorem}
\begin{proof}
We will first show that every $z \in C^{\text{AEL}}$ within the list decoding radius, i.e., $\Delta(z^R, \widetilde{y}^R) \leq (\delzero -\varepsilon)n$, belongs to the list $\mathcal{L}$ output by \cref{alg:AEL}.
Specifically, we just need to show that for each $z$ that we want to recover, we have $\widetilde{z}$ constructed from a measurable function on $\mathcal{B}_L$ such that $\Delta(z^L, \widetilde{z}^L) \leq \delone n$. 
This is because, if $\Delta(z^L, \widetilde{z}^L)\leq \delone n$, then we can use $C_{out}$'s unique decoder (after encoding $\widetilde{z}^L$ back to $[\qone]^n$ space) to recover $z$ from $\widetilde{z}$ in time $t$.
Now, from \cref{lem:ael-concentration}, we know that for such $z$, $x_z$ is $\eta$-concentrated over $\calB_{L|S_z}$.
So, the algorithm will enumerate over a $\widetilde{z}^L$ such that $\Delta(z^L, \widetilde{z}^L)\leq \eta|S_z|+\left|\overline{S_z}\right|$.
Vertices in $\overline{S_z}$ have $\leq (1-(\delzero-\varepsilon/2))d$ neighbors in $T_z$, but as $|T_z|\geq(1-(\delzero-\varepsilon))n$, we know that $\overline{S_z}$ must be small from \cref{clm:product-like-expansion}. 
Specifically, putting in the difference parameter $(1-(\delzero-\varepsilon)) - (1-(\delzero-\varepsilon/2)) = \varepsilon/2$, we get $\left|\overline{S_z}\right|\leq (\lambda/d)^2\cdot(4/\varepsilon^2) n$.
We can now substitute the values $\eta = (5\ell^2\gamma/\varepsilon)\cdot n/|S_z|$, $\gamma = \tfrac{\varepsilon\delone}{2^4\ell^2}$ and $\tfrac{\lambda}{d} < \tfrac{\varepsilon^2\delone^2}{2^{31}\ell^4}$ back in
$$\Delta(z^L, \widetilde{z}^L)\leq \eta|S_z|+\left|\overline{S_z}\right| \leq (5\delone/16 + \varepsilon^2\delone^4/2^{60})\cdot n\leq(\delone/3)n$$
Constructing the weak regularity decomposition for each $g_\alpha$ in $\alpha\in[\ell]^2$ takes $\ell^2\cdot\widetilde{\mathcal{O}}(2^{2^{\widetilde{\mathcal{O}}(1/\gamma^2)}} nd)$ time (\cref{thm:weak-reg}) in total.
As the total number of cut functions is $\ell^2\cdot\mathcal{O}(1/\gamma^2)$ which is a constant, constructing the factor (with $2^{\mathcal{O}(\ell^2/\gamma^2)}$ atoms) takes constant time. Also, while enumerating, we go through a constant $\ell^{2^{\mathcal{O}(\ell^2/\gamma^2)}}$ number of functions and invoke the unique decoder each time.
This means that in total \cref{alg:AEL} takes
$$
\mathcal{\widetilde{O}}\left(2^{2^{\mathcal{O}(1/\gamma^2)}}\ell^2d\cdot n + \ell^{2^{\mathcal{O}(\ell^2/\gamma^2)}}\cdot t\right)
$$
time which is $\mathcal{\widetilde{O}}_{\varepsilon,\ell}(n+t)$ if we suppress all the constants.
Substituting the value of $\gamma$ back in, the size of the list output by the algorithm is at most $\ell^{2^{\mathcal{O}(\ell^6/\varepsilon^2)}}$ as we go through, which directly
implies that $C^{\text{AEL}}$ is $(\delzero-\varepsilon, \ell^{2^{\mathcal{O}(\ell^6/\varepsilon^2)}})$ list decodable.
\end{proof}

\subsection{Instantiating the Code}
\label{subsec:ael-instantiate}
In this section, we will describe an appropriate (standard) instantiation of 
a family of AEL codes, which gives us the following parameters.
\begin{theorem}
    \label{thm:ael-main}
    Let $\qzero \ge 2$ be an integer and $\delzero \in (0, 1-1/\qzero)$.
    For any $\epsilon > 0$, there exists an explicit family of AEL codes $\{C_n \subseteq \F_q^n \}_{n}$ such that their
    \begin{enumerate}
        \item alphabet size is $q = \qzero^{\cal{O}(\varepsilon^{-20})}$,
        \item rate is $\geq 1-H_{\qzero}(\delzero)-\varepsilon$,
        \item design distance is $\delzero - \calO(\varepsilon^{10})$, and
        \item it is $(\delzero-\varepsilon, 2^{2^{\calO(\varepsilon^{-9})}})$ list decodable with a $2^{2^{\calO(\varepsilon^{-9})}} \cdot \widetilde{\mathcal{O}}(n)$ time (probabilistic) decoding algorithm 
    \end{enumerate}
\end{theorem}
\begin{proof}
    For any $d$ and $\qzero$, random codes $C_{in} \subseteq [\qzero]^d$ (with high probability) are list decodable up to $\delzero$ fraction of errors if the rate $\rzero = 1-H_{\qzero}(\delzero)-\epsilon_0$ with the list size $\ell = \mathcal{O}(1/\epsilon_0)$.
Also with high probability, random codes lie near the G.V. bound, so if the rate $\rzero = 1-H_{\qzero}(\delzero)-\epsilon_0$ then distance of the code is $\geq \delzero$.
Hence, using a union bound argument, we know that there exists a code $C_{in}$ with distance $\geq \delzero$ that is $(\delzero, \cal{O}(1/\varepsilon_0))$ list decodable with rate $\rzero \geq 1-H_{\qzero}(\delzero)-\varepsilon_0$.
As we take $d$ to be a fixed constant in our construction, we can find such a code by brute force.
The outer code $C_{out}$ is usually taken to be high-rate, and here we choose from the construction in \cite{1512415} which is a family of codes $\{C_{out,n}\}_n$ with parameters that can be set as: alphabet size $\qone = \qzero^{\rzero d}$, distance $\delone$ and rate $\rone \geq 1-\epsilon_1$.
This family of codes is uniquely decodable from $\delone^{dec} = \Omega(\epsilon_1^2)$ fraction of errors in linear time.

Using these parameters, the resulting code is $C^{\text{AEL}} = C(C_{in}, C_{out}, G)$ with (design) distance $\delta =\delzero - \tfrac{\lambda/d}{\delone}$ and rate 
$$r = \rzero\rone \geq (1-H_{\qzero}(\delzero)-\varepsilon_0)(1-\varepsilon_1)\geq 1-H_{\qzero}(\delzero)-\varepsilon_0 -\varepsilon_1$$
We know that $C^{\text{AEL}}$ is $(\delzero - \varepsilon, \ell^{2^{\cal{O}(\ell^6/\varepsilon^2)}})$ list decodable if $\tfrac{\lambda}{d}<\tfrac{\varepsilon^2(\delone^{dec})^2}{2^{31}\ell^4}$ from \cref{thm:AEL-listdec}.
To simplify the parameters, set $\varepsilon_0 = \varepsilon_1 = \varepsilon$ and replace $\varepsilon$ by $\varepsilon/2$.
This gives us that for any $\delzero>\varepsilon>0$, $C^{\text{AEL}}$ is $(\delzero-\varepsilon, 2^{2^{\calO(\varepsilon^{-9})}})$ list decodable with $r \geq 1-H_{\qzero}(\delzero)-\varepsilon$ if $\lambda/d < \cal{O}(\varepsilon^{10})$.
For these parameters, the design distance is $\delzero-\calO(\varepsilon^{10})$ for comparison.

We complete the construction by choosing graph $G$ from a family of bipartite expanders $\{G_n\}_n$ constructed by taking double cover of the graphs from \cite{Alon21} with a constant degree $d = \Theta(1/\varepsilon^{20})$ such that $\tfrac{\lambda}{d} < \cal{O}(1/\sqrt{d}) = \mathcal{O}(\varepsilon^{10})$.
With these parameters, the family of AEL codes $\{C(C_{in}, C_{out,n}, G_n)\}_n$ can be shown to be list decodable up to $\delzero-\varepsilon$ fraction of errors in $2^{2^{\calO(\varepsilon^{-9})}}\mathcal{\widetilde{O}}(n)$ or equivalently $\mathcal{\widetilde{O}}_\varepsilon(n)$ time by \cref{thm:AEL-listdec}.
\end{proof}

The inner code $C_{in}$ in the construction above can also be replaced by an explicit folded Reed-Solomon codes~\cite{GuruswamiR06}. For such codes with rate $\rzero$, the distance $\delzero = 1-\rzero-\varepsilon_0$, $\qzero = d^{\cal{O}_{\varepsilon_0}(1)}$ and the list size $\ell=O(1/\varepsilon_0)$ \cite{CZ25}. In~\cite{JMST25}, the code $C^{\text{AEL}}$ is proven to have the rate $r = \rzero(1-\varepsilon_1) \geq \rzero -\varepsilon_1$, design distance $\delta = 1-\rzero-\varepsilon_0 - \tfrac{\lambda/d}{\delone}$ such that it is $(\delta-\varepsilon, O(1/\epsilon_0))$ list decodable. Replacing $\delta$ with $1-R$ and using $\varepsilon_0=\varepsilon_1=\varepsilon$ with $\varepsilon$ replaced by $\varepsilon/3$ we retrieve the following corollary.

\begin{corollary}[\cref{cor:ael-opt-list} restated]
    Let $R \in (0,1)$.
    For any $\epsilon > 0$, there exists an explicit family of AEL codes $\{C_n \subseteq \F_q^n \}_{n}$ such that their
    \begin{enumerate}
        \item alphabet size is $q=2^{2^{\varepsilon^-\Theta(1)}}$,
        \item rate is $\geq R - \varepsilon$ while distance $\geq 1-R$ (near MDS), and
        \item it is $(1-R-\varepsilon, O(1/\epsilon))$ list decodable with a $\widetilde{\mathcal{O}}_{\epsilon}(n)$ time (probabilistic) decoding algorithm.
    \end{enumerate}
\end{corollary}

\section{List Decoding Tanner Codes}

Consider a Tanner code $C^\textnormal{TAN} = C(C_{0}, G) \subseteq \mathbb{F}_q^{nd}$ where $C_{0}$ is the base code and $G = (L, R, E)$ is a $(n, d, \lambda)$ bipartite expander.

In this section, we will show that there exists a family of Tanner codes that are list decodable with a radius close to the distance by a near-linear time algorithm (\cref{thm:tanner-main}).
This will be shown by first proving \cref{thm:tan-listdec} and then later describing an appropriate (standard) instantiation of codes in \cref{subsec:tan-instantiate}.
\begin{theorem}
    \label{thm:tan-listdec}
    If $C_{0}$ is $(\delta_{0}, \ell)$ list decodable, then the \text{Tanner} code $C^{\textnormal{TAN}}=C(C_{0}, G)$ is $(\delta_{0}(\delta_{0}-\varepsilon), \ell^{2^{\cal{O}(\ell^6/\varepsilon^6)}})$ list decodable for any constant $0<\varepsilon<\delta_{0}$ if $\tfrac{\lambda}{d} < \tfrac{\varepsilon^6}{2^{33}\ell^4}$.
    Moreover, this can be done in time $\widetilde{\cal{O}}_{\varepsilon, \ell}(n)$.
\end{theorem}

\subsection{Setting up the CSP}

\label{subsec:tan-csp-setup}
At each vertex $w \in L \cup R$, we can use the list decoding property of $C_{0}$ to retrieve a list of $\leq \ell$ candidate local codewords at most $\delta_{0} d$ away from the local view $\widetilde{y}_{w}$.
We will again set up a 2-CSP $\Ins_{\widetilde{y}} = (V,D,\Psi)$ with $2n$ variables, i.e., $V = \{x_w \;|\; w \in L\cup R\}$ and domain $D = [\ell]$. The locally recovered list elements at each vertex $w \in L\cup R$ are arbitrarily ordered like $c^{1}_w, c^{2}_w, \ldots, c^{\ell}_w$ so that they match with the domain values. If there are $<\ell$ candidate local codewords the remaining domain values represent some arbitrary entry in $C_{0}$.

The central point of difference from the setup for list decoding AEL is the nature of constraints used here. Each edge in the graph $G$ again represents a constraint in the CSP, i.e., $\Psi = \{\psi_e \; | \; e \in E\}$. But at an edge $(u, v) \in E$, the assignment $x_u = i, x_v = j$ satisfies the constraint iff $c^{i}_{u}$ and $c^{j}_{v}$  agree on the shared index, i.e., $\psi_{(u,v)} = \{(i,j) \; | \; c^{i}_{u, v} = c^{j}_{v, u}\}$.

\begin{algorithm}
\caption{List Decoding of Tanner code}
\label{alg:tanner}
\begin{algorithmic}
\Require $\widetilde{y} \in \mathbb{F}_q^E$
\Ensure $\mathcal{L} \subseteq C^\text{TAN}$
\\
\\
\begin{itemize}
    \item Construct a CSP $\Ins_{\widetilde{y}} = (V,D=[\ell],\Psi = \{\psi_e\;|\; e \in E\})$ as described in \cref{subsec:tan-csp-setup}
    \item For each $\alpha \in [\ell]^2$, obtain the weak regularity approximation $h_\alpha$ of $g_\alpha : [n]^2 \to \{0,1\}$, as defined below, with $\gamma = \frac{\varepsilon^3}{2^5\ell^2}$
\[g_\alpha(u, v) = \begin{cases}
    1 &\quad \text{if }\,(u,v) \in E \,\text{ and }\,\alpha \in \psi_{(u,v)} \\
    0 &\quad \text{otherwise}
\end{cases}\]
    \item Let $h_\alpha = \sum_{t=1}^{p_\alpha} k_{\alpha,t} \cdot \mathbf{1}_{\mathcal{S}_{\alpha, t}} \otimes \mathbf{1}_{\mathcal{T}_{\alpha, t}}$. Then, compute the factor $\mathcal{B}_R$  determined by the set of cut functions $\{\mathbf{1}_{\mathcal{T}_{\alpha, t}} \;|\; \alpha\in[\ell]^2, t\in[p_\alpha]\}$
    \item $\mathcal{L} \gets \{\}$
    \item For every $\mathcal{B}_R$ measurable function $\widetilde{x}:R \to[\ell]$ 
    \begin{itemize}
        \item[-] Obtain $\widetilde{z} \in \mathbb{F}_q^E$ such that
        for each $(u,v)\in E$, $\widetilde{z}_{u,v} = c^{\widetilde{x}(v)}_{v,u}$
        \item[-] $\mathcal{L}_{\widetilde{z}} \gets \textsc{CustomDecode}(\widetilde{z})$
        \item[-] For each $z \in \mathcal{L}_{\widetilde{z}}$, if $\Delta(z, \widetilde{y}) \leq \delta_{0}(\delta_{0} - \varepsilon) nd$
        \begin{itemize}
            \item[-] $\mathcal{L}\gets\mathcal{L} \cup \{z\}$
        \end{itemize}
    \end{itemize}
    \item Return $\mathcal{L}$
\end{itemize}
\end{algorithmic}
\end{algorithm}

\begin{algorithm}
\caption{\textsc{CustomDecode}}
\label{alg:cust-dec}
\begin{algorithmic}
\Require $\widetilde{z} \in \mathbb{F}_q^E$
\Ensure $\mathcal{L}_{\widetilde{z}} \subseteq C^\text{TAN}$
\\
\\
\begin{itemize}
    \item Construct a CSP $\Ins_{\widetilde{z}} = (V,D=[\ell],\Psi = \{\psi_e\;|\; e \in E\})$ as described in \cref{subsec:ael-csp-setup}
    \item For each $\alpha \in [\ell]^2$, obtain the weak regularity approximation $h_\alpha$ of $g_\alpha : [n]^2 \to \{0,1\}$, as defined below, with $\gamma = \frac{\varepsilon^3}{2^5\ell^2}$
\[g_\alpha(u, v) = \begin{cases}
    1 &\quad \text{if }\,(u,v) \in E \,\text{ and }\,\alpha \in \psi_{(u,v)} \\
    0 &\quad \text{otherwise}
\end{cases}\]
    \item Let $h_\alpha = \sum_{t=1}^{p_\alpha} k_{\alpha,t} \cdot \mathbf{1}_{\mathcal{S}_{\alpha, t}} \otimes \mathbf{1}_{\mathcal{T}_{\alpha, t}}$. Then, compute the factor $\mathcal{B}_L$  determined by the set of cut functions $\{\mathbf{1}_{\mathcal{S}_{\alpha, t}} \;|\; \alpha\in[\ell]^2, t\in[p_\alpha]\}$
    \item $\mathcal{L}_{\widetilde{z}} \gets \{\}$
    \item For every $\mathcal{B}_L$ measurable function $x':L \to[\ell]$ 
    \begin{itemize}
        \item[-] Obtain $z' \in \mathbb{F}_q^E$ such that
        for each $(u,v)\in E$, $z'_{u,v} = c^{x'(u)}_{u,v}$
        \item[-] $z \gets \textsc{UniqueDecode}(z')$
        \item[-] If $z \neq \textsc{Reject}$
        \begin{itemize}
            \item[-] $\mathcal{L}_{\widetilde{z}}\gets\mathcal{L}_{\widetilde{z}} \cup \{z\}$
        \end{itemize}
    \end{itemize}
    \item Return $\mathcal{L}_{\widetilde{z}}$
\end{itemize}
\end{algorithmic}
\end{algorithm}

\subsection{Weak Regularity Decomposition}
We will now use the weak regularity decomposition to analyze the structure of a relevant subportion of the CSP.
Note that with $g_\alpha$ as defined in \cref{alg:tanner}, the number of satisfied constraints of the form $\{\psi_e\;|\;e \in E(S, T)\;\}$ for any $S\subseteq L$ and $T\subseteq R$ can be approximated as follows.
\begin{align*}
    \text{val}_{S,T}(\Ins_{\widetilde{y}}, x) &= \sum_{\alpha=(\alpha_L, \alpha_R) \in [\ell]^2} \left\langle g_\alpha, \mathbf{1}_{x,{\alpha_L}} \cdot \mathbf{1}_{S} \otimes \mathbf{1}_{x,{\alpha_R}} \cdot \mathbf{1}_{T} \right\rangle \\
    &= \sum_{\alpha=(\alpha_L, \alpha_R) \in [\ell]^2} \left\langle h_\alpha, \mathbf{1}_{x,{\alpha_L}} \cdot \mathbf{1}_{S} \otimes \mathbf{1}_{x,{\alpha_R}} \cdot \mathbf{1}_{T} \right\rangle \;\pm\; \ell^2 \cdot \gamma \cdot nd
\end{align*}
Assuming $h_\alpha$ has the form $\sum_{t=1}^{p_\alpha} k_{\alpha,t} \cdot \mathbf{1}_{\mathcal{S}_{\alpha, t}} \otimes \mathbf{1}_{\mathcal{T}_{\alpha, t}}$, the factor $\mathcal{B}_L$ is determined by the set $\{\mathbf{1}_{\mathcal{S}_{\alpha, t}} \;|\; \alpha\in[\ell]^2, t\in[p_\alpha]\}$ and $\mathcal{B}_R$ by $\{\mathbf{1}_{\mathcal{T}_{\alpha, t}} \;|\; \alpha\in[\ell]^2, t\in[p_\alpha]\}$.
So,
\begin{align*}
    \text{val}_{S,T}(\Ins_{\widetilde{y}}, x) &= \sum_{\alpha \in [\ell]^2} \sum_{t=1}^{p_\alpha} k_{\alpha,t} \left\langle \mathbf{1}_{\mathcal{S}_{\alpha, t}} \otimes \mathbf{1}_{\mathcal{T}_{\alpha, t}}, \mathbf{1}_{x,{\alpha_L}} \cdot \mathbf{1}_{S} \otimes \mathbf{1}_{x,{\alpha_R}} \cdot \mathbf{1}_{T} \right\rangle \;\pm\; \ell^2 \gamma nd \\
    &= \sum_{\alpha, t} k_{\alpha,t} \left\langle \mathbf{1}_{\mathcal{S}_{\alpha, t}} , \mathbf{1}_{x,{\alpha_L}} \right\rangle_S\left\langle \mathbf{1}_{\mathcal{T}_{\alpha, t}}, \mathbf{1}_{x,{\alpha_R}} \right\rangle_T \;\pm\; \ell^2 \gamma nd \\
    &=\sum_{\alpha, t} k_{\alpha,t} \left\langle \mathbf{1}_{\mathcal{S}_{\alpha, t}} , \E[\mathbf{1}_{x,{\alpha_L}} \,| \,\mathcal{B}_{L|S}] \right\rangle_S\left\langle \mathbf{1}_{\mathcal{T}_{\alpha, t}}, \E[\mathbf{1}_{x,{\alpha_R}} \,|\, \mathcal{B}_{R|T}] \right\rangle_T \;\pm\; \ell^2 \gamma nd
\end{align*}
The last step is possible because the set of functions $\{\mathbf{1}_{\mathcal{S}_{\alpha, t}}\}_{\alpha,t}$ and $\{\mathbf{1}_{\mathcal{T}_{\alpha, t}}\}_{\alpha,t}$ is measurable on $\mathcal{B}_{L|S}$ and $\mathcal{B}_{R|T}$ respectively.
This implies that $\text{val}_{S,T}(\Ins_{\widetilde{y}}, x)$ is dependent only on the distribution induced by the assignment function $x$ on the factor $\mathcal{B}_{L|S}$ and $\mathcal{B}_{R|T}$.

\subsection{Concentration over Factors}
For any $z$ that we want to recover, define the subsets $S_z\subseteq L$ and $T_z \subseteq R$ as
\begin{align*}
    &T^*_z = \{v \;|\; v \in R, \Delta(z_v, \widetilde{y}_v)\leq\delta_{0} d\} \\
    &T_z = \{v \;|\; v\in T^*_z, |N(v)
    \setminus S_z| \leq (\delta_{0}-\varepsilon/2)d\} \quad S_z = \{u \;|\; u \in L, \Delta(z_u, \widetilde{y}_u)\leq\delta_{0} d\}
\end{align*}
The assignment function $x_z : V \to [\ell]$ is
$$
    x_z(w) = \begin{cases}
    i &\quad \text{if } w \in S_z\cup T_z \text{ and } z_{w} = c^i_w\\
    \text{arbitrary} &\quad \text{otherwise}
    \end{cases}
$$
\begin{claim}
    \label{clm:full-sat-tan}
    $\textnormal{val}_{S_z, T_z}(\Ins_{\widetilde{y}}, x_z) = |E(S_z, T_z)|$
\end{claim}
\begin{proof}
    For any $w \in S_z \cup T_z$, $z_w$ belongs to the locally decoded list. So, for any $(u,v) \in E(S_z, T_z)$, $c^{x_z(u)}_u = z_u$ and $c^{x_z(v)}_v = z_v$. As $c^{x_z(u)}_{u,v} = z_{u,v} = z_{v,u} = c^{x_z(v)}_{v,u}$, any constraint associated to such an edge is satisfied.
\end{proof}
\begin{lemma}
    \label{lem:tanner-concentration}
    If $z \in C^{\textnormal{TAN}}$ and $\Delta(z, \widetilde{y}) \leq \delta_{0}(\delta_{0} - \varepsilon)nd$, then $x_z$ is $\eta$-concentrated on $\mathcal{B}_{R|T_z}$ for $\eta = (5\ell^2\gamma/\varepsilon)\cdot n/|T_z|$.
\end{lemma}
\begin{proof}    
    The proof here is similar to the proof of \cref{lem:ael-concentration}.
    So, for the sake of being concise we will refer to parts of \cref{lem:ael-concentration}'s proof here.
    As we need to prove concentration on the factor $\mathcal{B}_{R|T_z}$, consider (by the method of contradiction) what happens if $x_z$ is not concentrated.
    The assignment function is kept the same on all the variables except the ones in $T_z$, where it is moved around inside the partitions according to the procedure in \cref{lem:ael-concentration}.
    This has the effect of changing the assignment at $\geq \eta|T_z|$ vertices in $T_z$.
    Now notice that at any such changed vertex $v \in T_z$, the assignment $x_z$ satisfies all the connected constraints $(u,v)\in E(S_z, T_z)$  (\cref{clm:full-sat-tan}).
    Any change to the assignment $x_z(v) = j \neq j'$ here will have an effect that for $\geq \delta_{0} d$ edges $(u,v)\in E$, $c^{j}_{v, u} \neq c^{j'}_{v,u}$.
    But, by definition, for all $v\in T_z$, only $\leq (\delta_{0}-\varepsilon/2)d$ edges are not in $E(S_z, T_z)$.
    So, in total, $\geq (\varepsilon/2)d\cdot\eta|T_z|$ constraints in $E(S_z, T_z)$ will become unsatisfied by the specified change (as the constraints are of the form $c^{x_z(u)}_{u,v} = c^{j}_{v, u}$ and $c^{j}_{v, u} \neq c^{j'}_{v,u}$).
    This leads to a contradiction because the difference in the number of satisfied constraints by these two assignments should be $\geq 2.5\ell^2\gamma\cdot nd$ due to the above analysis, but $\leq 2\ell^2\gamma\cdot nd$ from the weak regularity decomposition.
\end{proof}
\subsection{Covering relevant codewords}
We will now prove that due to the concentrated structure shown above, the enumeration in \cref{alg:tanner} covers all the relevant codewords.

\begin{theorem}[\cref{thm:tan-listdec} restated]
    If $C_{0}$ is $(\delta_{0}, \ell)$ list decodable, then the \text{Tanner} code $C^{\textnormal{TAN}}=C(C_{0}, G)$ is $(\delta_{0}(\delta_{0}-\varepsilon), \ell^{2^{\cal{O}(\ell^6/\varepsilon^6)}})$ list decodable for any constant $0<\varepsilon<\delta_{0}$ if $\tfrac{\lambda}{d} < \tfrac{\varepsilon^6}{2^{33}\ell^4}$.
    Moreover, this can be done in time $\widetilde{\cal{O}}_{\varepsilon, \ell}(n)$.
\end{theorem}
\begin{proof}
First, we will show that every $z \in C^\text{TAN}$ such that $\Delta(z, \widetilde{y})\leq \delta_{0}(\delta_{0} -\varepsilon)nd$ is in the list output by \cref{alg:tanner}.
This means that the maximum size of the output list explicitly bounds the number of codewords within the list decoding radius.
So, it is sufficient to prove that \cref{alg:tanner} can run in $\widetilde{\cal{O}}_{\varepsilon, \ell}(n)$ and that it outputs a list with size $\leq \ell^{2^{\cal{O}(\ell^6/\varepsilon^6)}}$.
To this end, note that from \cref{lem:tanner-concentration} we know that for all the $z$'s we want to recover, $x_z$ is concentrated on $\mathcal{B}_{R|T_z}$.
This means that $x_z$ differs from a measurable assignment $\widetilde{x}_z$ (which gets enumerated over in the algorithm) at most at $\eta|T_z|+\left|\overline{T_z}\right|$ vertices on $R$.
To analyze the sizes of these sets, observe that 
$$\overline{T_z} = (\overline{T_z^*}\setminus T_z) \cup (T_z^*\setminus T_z) =
\overline{T_z^*} \cup (T_z^*\setminus T_z)$$
For each vertex $v\in\overline{T^*_z}, \;\Delta(z_v, \widetilde{y}_v)> \delta_{0} d$ but $\Delta(z, \widetilde{y})\leq \delta_{0}(\delta_{0}-\varepsilon)nd$. So, $\left|\overline{T^*_z}\right| \leq (\delta_{0}-\varepsilon)n$.
By the same reasoning $\left|\overline{S_z}\right|\leq (\delta_{0}-\varepsilon)n$. But as vertices in $T_z^*\setminus T_z$ have less than $(1-(\delta_{0}-\varepsilon/2))d$ neighbors in $S_z$, combining it with \cref{clm:product-like-expansion} gives us that $|T_z^*\setminus T_z| \leq (\lambda/d)^2\cdot(4/\varepsilon^2) n$. Putting all constants $\eta = (5\ell^2\gamma/\varepsilon)\cdot n/|T_z|$, $\gamma = \tfrac{\varepsilon^3}{2^5\ell^2}$ and $\tfrac{\lambda}{d} < \tfrac{\varepsilon^6}{2^{33}\ell^4}$ back into the equation, we get
\begin{align*}
    \eta|T_z|+\left|\overline{T_z}\right| &\leq \eta |T_z| + \left|\overline{T_z^*}\right| + \left|T_z^*\setminus T_z\right| \\
    &\leq (5\ell^2\gamma/\varepsilon) n + (\delta_{0}-\varepsilon)n + (\lambda/d)^2\cdot(4/\varepsilon^2)n \\
    &\leq (5\varepsilon^2/32)n + (\delta_{0}-\varepsilon)n + (\varepsilon^{10}/2^{64})n = (\delta_{0}-\varepsilon/2)n
\end{align*}
This means that for $\geq \left(1-(\delta_{0}-\varepsilon/2)\right)n$ vertices $v \in R$, $\widetilde{x}_z(v) = x_z(v)$ and hence $z_v = c^{x_z(v)}_{v} = c^{\widetilde{x}_z(v)}_{v} = \widetilde{z}_v$.
If $\widetilde{z}$ is constructed from $\widetilde{x}_z$ as done in \cref{alg:tanner}, establishing the following claim completes the proof of the fact that it recovers all the relevant $z$'s.
\begin{claim}
    $z \in \textnormal{\textsc{CustomDecode}}(\widetilde{z})$
\end{claim}
\begin{proof}
    The CSP setup in \cref{alg:cust-dec} is the one described in \cref{subsec:ael-csp-setup}, where a constraint associated to $(u,v)\in E$ is satisfied by an assignment $x$ if $c^{x(u)}_{u,v} = \widetilde{z}_{u,v}$.
    In this setup, we can again define the subsets $S_z\subseteq L$ and $T_z \subseteq R$ as
\begin{align*}
    &S_z^* = \{u \;|\; u \in L,\Delta(z_u, \widetilde{z}_u)\leq \delta_{0} d\} \\
    &S_z = \{u \;|\; u \in S_z^*, |N(u) \setminus T_z|\leq(\delta_{0}-\varepsilon/4)d\} \quad\quad T_z = \{v \;|\; v \in R, z_v = \widetilde{z}_v\}
\end{align*}
And the assignment function $x_z : V \to [\ell]$ as
$$
    x_z(w) = \begin{cases}
    i &\quad \text{if } w \in S_z \text{ and } z_w = c^i_w\\
    \text{arbitrary} &\quad \text{otherwise}
    \end{cases}
$$
Following the same reasoning, we know that $x_z$ is $\eta$-concentrated on $\mathcal{B}_{L|S_z}$.
So, while enumerating over measurable assignment functions, we will go through a $x'_{z}$ which only differs from $x_z$ on $\leq \eta|S_z| + \left|\overline{S_z}\right|$ vertices in $L$.
Now, we know that $|T_z|\geq(1-(\delta_{0}-\varepsilon/2))n$ from the previous discussion.
But as $\overline{S_z}$ have $\leq (1-(\delta_{0}-\varepsilon/4))d$ neighbors in $T_z$ it must be that $\left|\overline{S_z}\right|\leq (\lambda/d)^2\cdot(16/\varepsilon^2) n$ from \cref{clm:product-like-expansion} (as the difference in parameters here is $(1-(\delta_{0}-\varepsilon/2))-(1-(\delta_{0}-\varepsilon/4)) = \varepsilon/4$).
Substituting the values $\eta = (5\ell^2\gamma/\varepsilon)\cdot n/|S_z|$, $\gamma = \tfrac{\varepsilon^3}{2^5\ell^2}$ and $\tfrac{\lambda}{d} < \tfrac{\varepsilon^6}{2^{33}\ell^4}$, we see that
$$\eta|S_z| + \left|\overline{S_z}\right| \leq (5\varepsilon^2/32 + \varepsilon^{10}/2^{62})\cdot n \leq (\varepsilon^2/6)n$$
This means that the $z'$ recovered from $x'_{z}$ is such that $\Delta(z,z')\leq (\varepsilon^2/6)nd \leq (\delta_{0}^2/6)nd$ which falls inside the unique decoding threshold. So, $z = \textsc{UniqueDecode}(z')$ and hence $z \in \mathcal{L}_{\widetilde{z}}$.
\end{proof}
To complete the proof of \cref{thm:tan-listdec}, we discuss the running time of \cref{alg:tanner} and the list size. 
Constructing the weak regularity decomposition for each $g_\alpha$ in $\alpha\in[\ell]^2$ takes $\ell^2\cdot\widetilde{\mathcal{O}}(2^{2^{\widetilde{\mathcal{O}}(1/\gamma^2)}} nd)$ time (\cref{thm:weak-reg}).
Enumerating over measurable functions, we invoke the custom decoder a constant $\ell^{2^{\mathcal{O}(\ell^2/\gamma^2)}}$  number of times.
As the custom decoder again builds a weak regularity decomposition (which takes $\ell^2\cdot\widetilde{\mathcal{O}}(2^{2^{\widetilde{\mathcal{O}}(1/\gamma^2)}} \cdot nd)$ time) and invokes the unique decoder of $C_{0}$ (which takes $\mathcal{O}(n)$ time to run each time) a constant  $\ell^{2^{\mathcal{O}(\ell^2/\gamma^2)}}$ number of times, it runs in $\mathcal{\widetilde{O}}(\ell^{2^{\mathcal{O}(\ell^2/\gamma^2)}} \cdot n)$ time overall.
This means that overall, \cref{alg:tanner} takes $\mathcal{\widetilde{O}}(\ell^{2^{\mathcal{O}(\ell^2/\gamma^2)}} \cdot n)$ time to run too, which can be expressed as $\mathcal{\widetilde{O}}_{\varepsilon, \ell}(n)$ if we suppress the constants.
The output list size is at most $\ell^{2^{\mathcal{O}(\ell^2/\gamma^2)}}\cdot \ell^{2^{\mathcal{O}(\ell^2/\gamma^2)}}$ which is $\ell^{2^{\mathcal{O}(\ell^6/\varepsilon^6)}}$ asymptotically if we substitute the value of $\gamma$.
\end{proof}
\subsection{Instantiating the Code}\label{subsec:tan-instantiate}

We now proceed to concrete instantiations of Tanner codes.

\begin{theorem}
    \label{thm:tanner-main}
    Let $q \ge 2$ be a prime power and $\delta_{0} \in (0, 1-1/q)$.
    For any $\epsilon > 0$, there exists an explicit family of Tanner codes $\{C_n \subseteq \F_q^n \}_{n}$ such that their
    \begin{enumerate}
        \item rate is $\geq 1-2H_q(\delta_{0}) - \varepsilon$,
        \item design distance is $\delta_0^2 - \calO(\varepsilon^{10})$, and
        \item it is $(\delta_{0}^2-\varepsilon, 2^{2^{\calO(\varepsilon^{-13})}})$ list decodable with a $2^{2^{\calO(\varepsilon^{-13})}} \cdot \widetilde{\mathcal{O}}(n)$ time (probabilistic) decoding algorithm.
    \end{enumerate}
\end{theorem}
\begin{proof}
For any prime power $q$ and block length $d$, we know that with high probability a random linear code $C_{0} \subseteq [q]^d$ is list decodable up to $\delta_{0} \in (0,1-1/q)$ fraction of errors if rate $r_{0} = 1-H_{q}(\delta_0)-\epsilon_0$ with list size $\ell = \mathcal{O}(1/\epsilon_0)$ \cite{DBLP:journals/corr/abs-1001-1386}.
It is also known that random linear codes lie near the G.V. bound, so with rate $r_{0} = 1-H_{q}(\delta_{0})-\epsilon_0$  the distance is $\geq \delta_{0}$.
By a union bound, codes with both these properties must exist. Thus, we can brute force the constant-sized search space to find a local code with distance $\geq \delta_{0}$ that is $(\delta_{0}, \mathcal{O}(1/\varepsilon_0))$ list decodable with rate $r_{0} \geq 1-H_q(\delta_{0}) - \varepsilon_0$.

In terms of the parameters defined above, the Tanner code $C^{\text{TAN}} = C(C_{0}, G)$ has rate $r \geq 2r_{0} - 1 \geq 1-2H_q(\delta_0)-2\varepsilon_0$ and (design) distance $\delta =\delta_{0}(\delta_{0} - \lambda/d)$. 
We know that $C^{\text{TAN}}$ can be list decoded from 
$\delta_{0}(\delta_{0} - \varepsilon) \geq \delta_{0}^2 -\varepsilon$ with a list size of $\ell^{2^{\mathcal{O}({\ell^6/\varepsilon^{6}})}}$ if $\tfrac{\lambda}{d} < \tfrac{\varepsilon^6}{2^{33}\ell^4}$.
Setting $\varepsilon_0 = \varepsilon$ and then replacing $\varepsilon$ with $\varepsilon/2$ gives us a cleaner form of the statement in \cref{thm:tan-listdec}: for any $\varepsilon>0$, $C^{\text{TAN}}$ is $(\delta_{0}^2-\varepsilon, 2^{2^{\calO(\varepsilon^{-13})}})$ list decodable if $\lambda/d < \cal{O}(\varepsilon^{10})$.
Note that in these terms, the design distance is $\delta_0^2 - \calO(\varepsilon^{10})$.

To finish our construction, we need to find an appropriate family of graphs which have expansion better than $\cal{O}(\varepsilon^{10})$.
For this, we use the construction from \cite{Alon21}, which achieves $\lambda/d = \cal{O}(1/\sqrt{d})$ for every $n$ and take its double cover to get a bipartite version $\{G_n\}_n$.
We use $d = \Theta(1/\varepsilon^{20})$ to satisfy the above expansion requirements.
With these parameters, the family of Tanner codes $\{C(C_{0}, G_n)\}_n$ can be shown to be list decodable up to $\delta_{0}^2-\varepsilon$ fraction of errors in $2^{2^{\calO(\varepsilon^{-13})}}\mathcal{\widetilde{O}}(n)$ or $\mathcal{\widetilde{O}}_\varepsilon(n)$ time by \cref{thm:tan-listdec}.
\end{proof}

\bibliographystyle{alpha}
\bibliography{macros,references}

%\appendix
%\input{appendix}

\end{document}